\documentclass[
aps,
prx,
10pt,
twocolumn,
groupedaddress,
longbibliography,
nofootinbib,
floatfix
]{revtex4-2}

\usepackage{amsmath,amssymb,amsthm,mathtools}
\usepackage{graphicx}
\usepackage{booktabs}
\usepackage{array}
\usepackage{bm}
\usepackage{microtype}
\usepackage{xcolor}
\usepackage{tikz}
\usepackage[ruled,noend]{algorithm2e}
\SetAlgoNlRelativeSize{0}
\SetAlgoSkip{}
\DontPrintSemicolon
\SetAlgorithmName{ALGORITHM}{ALGORITHM}{List of Algorithms}
\SetAlCapSty{}
\SetAlCapNameFnt{\footnotesize}
\SetAlCapFnt{\footnotesize}

\makeatletter
\renewcommand{\algocf@capseparator}{.~}
\makeatother
\usepackage{enumitem}
\usepackage{silence}
\usepackage{hyperref}
\hypersetup{colorlinks=true,linkcolor=blue,citecolor=blue,urlcolor=blue,hypertexnames=false}
\usepackage{orcidlink}
\newcommand{\nSemanticTotal}{160}

\newcommand{\nBBSeventyTwoPoints}{132}
\newcommand{\nBBNinetyPoints}{6}
\newcommand{\nBBOneOhEightPoints}{6}
\newcommand{\nBBOneFortyFourPoints}{16}
\newcommand{\baselineCorrelationN}{101}
\newcommand{\baselineSpearmanRho}{0.893}
\newcommand{\baselineSpearmanP}{3.62\times10^{-36}}
\newcommand{\bbSeventyTwoCrossMonoLER}{0.2532}
\newcommand{\bbSeventyTwoCrossMonoLERLo}{0.2470}
\newcommand{\bbSeventyTwoCrossMonoLERHi}{0.2595}
\newcommand{\bbSeventyTwoCrossBiLERHi}{0.000368}

\newcommand{\bbSeventyTwoCrossSupportCrossings}{15}
\newcommand{\bbSeventyTwoCrossMatching}{3}
\newcommand{\bbSeventyTwoCrossDeffMono}{3}
\newcommand{\bbSeventyTwoCrossDeffBi}{6}
\newcommand{\bbSeventyTwoPowerMonoLER}{0.2868}
\newcommand{\bbSeventyTwoPowerMonoLERLo}{0.2839}
\newcommand{\bbSeventyTwoPowerMonoLERHi}{0.2897}
\newcommand{\bbSeventyTwoPowerBiLER}{0.0581}
\newcommand{\bbSeventyTwoPowerBiLERLo}{0.0566}
\newcommand{\bbSeventyTwoPowerBiLERHi}{0.0596}
\newcommand{\bbSeventyTwoPowerLAFourLER}{0.2054}
\newcommand{\bbSeventyTwoPowerLAFourLERLo}{0.2025}
\newcommand{\bbSeventyTwoPowerLAFourLERHi}{0.2082}
\newcommand{\bbSeventyTwoPowerLERRatio}{4.938}
\newcommand{\bbSeventyTwoMonoExposure}{0.0246}
\newcommand{\bbSeventyTwoBiExposure}{0.0150}
\newcommand{\bbSeventyTwoPowerAggProbMono}{0.0522}
\newcommand{\bbSeventyTwoPowerAggProbBi}{0.0081}
\newcommand{\bbOneFortyFourPowerMonoLER}{0.3220}
\newcommand{\bbOneFortyFourPowerMonoLERLo}{0.3125}
\newcommand{\bbOneFortyFourPowerMonoLERHi}{0.3316}
\newcommand{\bbOneFortyFourPowerBiLER}{0.0031}
\newcommand{\bbOneFortyFourPowerBiLERLo}{0.0025}
\newcommand{\bbOneFortyFourPowerBiLERHi}{0.0038}
\newcommand{\bbOneFortyFourPowerLERRatio}{103.9}
\newcommand{\bbSeventyTwoPureLFamilyCount}{36}
\newcommand{\bbSeventyTwoLAMaxExposureMono}{0.0303}
\newcommand{\bbSeventyTwoLAMaxExposureLA}{0.0224}
\newcommand{\bbSeventyTwoLAMaxExposureBi}{0.0170}
\newcommand{\bbSeventyTwoLAImprovementPct}{26.11}
\newcommand{\bbSeventyTwoBiImprovementPct}{44.03}

\newtheorem{theorem}{Theorem}
\newtheorem{proposition}{Proposition}
\newtheorem{lemma}{Lemma}
\newtheorem{corollary}{Corollary}
\theoremstyle{definition}
\newtheorem{definition}{Definition}
\newtheorem{remark}{Remark}
\newtheorem{assumption}{Assumption}

\DeclareMathOperator{\supp}{supp}
\DeclareMathOperator{\wt}{wt}

\DeclareMathOperator{\kerop}{ker}

\DeclareMathOperator*{\argmin}{arg\,min}

\newcommand{\F}{\mathbb{F}_2}
\newcommand{\I}{\hat{\mathbb{I}}}
\newcommand{\Ew}{\mathcal{W}}
\newcommand{\Pgroup}{\mathbb{P}}

\newcommand{\RX}{\mathcal{R}_X}
\newcommand{\TL}{T_L}

\newcommand{\bb}[1]{\ensuremath{\left[\!\left[#1\right]\!\right]}}
\newcommand{\Hx}{\hat H_{\times}}
\newcommand{\Hxt}{\hat H_{\times}^{(t)}}

\begin{document}

\title{Geometry-induced correlated noise in qLDPC syndrome extraction}

\author{Angelo Di Bella\,\orcidlink{0009-0005-5593-2993}}
\affiliation{Cavendish Laboratory, University of Cambridge,\\JJ Thomson Ave, Cambridge CB3 0HE, United Kingdom}
\email{ad2395@cam.ac.uk}

\begin{abstract}
Routed geometry is a device-level choice in a fixed syndrome-extraction circuit. Two embeddings of the same code can set different physical separations between gate blocks active in the same time step, and these separations control the residual coupling between those blocks. We derive how this choice shapes the leading correlated-fault structure of the effective data channel, and we test the consequences at circuit level. Starting from a geometry-conditioned interaction Hamiltonian on disjoint blocks within one tick, we obtain a retained data channel of single and pair faults for bivariate-bicycle codes, with a truncation error controlled by the per-tick coupling strength. Two geometry metrics emerge. In the combinatorial limit, a matching argument on the logical support reduces the effective fault weight on that support. For strictly positive kernels, once every support pair contributes somewhere in the schedule, the induced support graph becomes complete. At that point the matching-number reduction is exhausted, and the embedding-dependent quantity is the total retained pair weight on the support, which we call the \emph{weighted exposure}. Circuit-level Monte Carlo on the \bb{72,12,6} and \bb{144,12,12} benchmarks shows that a biplanar layout, with the schedule split across two routing planes, suppresses the geometry penalty incurred by the monomial layout in a single plane. On the BB72 baseline set of \baselineCorrelationN{} operating points, the reference-support weighted exposure is strongly correlated with the observed logical error rate (Spearman $\rho_\mathrm{S}=\baselineSpearmanRho$) in the tested window. A logical-aware two-swap local search over single-layer embeddings on BB72 reduces the worst-case family exposure by \bbSeventyTwoLAImprovementPct\% and lowers the logical error rate across the tested power-law window.
\end{abstract}

\maketitle

\section{Introduction}
\label{sec:intro}

Quantum low-density parity-check (qLDPC) codes---stabilizer codes whose parity-check matrices have bounded row and column weight---combine sparse stabilizer structure with finite-rate families, practical syndrome-extraction protocols, and efficient decoders. Bivariate-bicycle (BB) codes are concrete memory benchmarks with explicit layer decompositions and finite-length data~\cite{calderbank_shor_1996,steane_1996,breuckmann_eberhardt_2021,bravyi_memory_2024}. Long-range-coupled BB experiments make routing geometry a device-level variable rather than a schematic choice~\cite{wang_demo_2025}. Existing architecture work usually optimizes connectivity, routing depth, or hardware integration. It does not, however, carry a fixed code family and extraction schedule through the full chain from routed geometry to correlated fault model to logical performance. Two embeddings of the same Tanner graph---the bipartite graph connecting data qubits to parity checks---can set different physical separations between the gate blocks active in any one time step. Those separations alter the rate of correlated pair events, the effective fault weight on logical operators, and the noise model the decoder is implicitly assuming, all without changing the stabilizer algebra.

Three independent strands of recent work motivate the question and supply the ingredients to answer it. The first is a circuit-centric view of fault tolerance, in which the static parity-check matrices alone do not determine logical performance: residual-error metrics for syndrome-extraction design~\cite{strikis_sec_2026} and spacetime-code constructions~\cite{aitchison_beri_2025,pesah_spacetime_2025} demonstrate that schedule and routing choices change which fault patterns dominate. The second is the long-range-noise tradition initiated by Aharonov, Kitaev, and Preskill (AKP)~\cite{akp_2006}, which gives the asymptotic criterion under which pairwise long-range Hamiltonian noise remains compatible with fault tolerance, and which has been complemented by surface-code crosstalk studies~\cite{zhou_crosstalk_2025} showing that correlated couplings shift logical thresholds in finite-size simulations. The third is device-level evidence that residual inter-block coupling profiles in superconducting hardware depend on physical layout, couplers, shielding, and calibration~\cite{pettersson_zz_2024,aguila_flux_2025}; together these microscopic analyses justify treating routed separation as a control variable on the correlated-fault budget. Recent qLDPC decoding work, including correlation-aware extensions of belief propagation with ordered statistics decoding (BP+OSD)~\cite{roffe_bp_osd_2020,ducrest_capbp_2024,hillmann_lsd_2025,mueller_relaybp_2025,maan_correlated_2026,sahay_matching_2026}, indicates that geometry-induced correlations can in principle be folded into production decoders, so the question of how much logical-performance signal a single fixed decoder still sees is well posed.

What this literature does not yet provide is a single chain that takes a fixed code family and extraction schedule, derives from a microscopic interaction model the resulting correlated-fault structure under different routed embeddings, and tests the difference at circuit level for a fixed decoder. The present work supplies that chain for bivariate-bicycle codes on a routed Calderbank--Shor--Steane (CSS) extraction circuit. Under the same-tick model developed below, in which the inter-block interaction term dominates, routed geometry changes the retained pair-fault structure, and in turn the effective fault weight on logical supports and the finite-coupling logical performance. Table~\ref{tab:related-work} locates the present work against the immediate precedents.

\begin{table}[!htbp]
\caption{Comparison with recent qLDPC architecture work. Each column indicates whether the reference treats routed geometry, correlated noise, a circuit-level analysis, and/or a geometry-derived metric.}\label{tab:related-work}
\footnotesize
\renewcommand{\arraystretch}{1.2}%
\begin{ruledtabular}
\begin{tabular}{ccccc}
Reference & Geometry & Corr.\ noise & Circuit & Metric \\
\colrule
\onlinecite{bravyi_memory_2024} & layer split & & threshold & \\
\onlinecite{berthusen_local_2025} & 2D-local & & threshold & connectivity \\
\onlinecite{mathews_multilayer_2025} & multilayer & & & route length \\
\onlinecite{wang_demo_2025} & hardware & device & experiment & \\
\onlinecite{strikis_sec_2026} & & & residual error & \\
\onlinecite{zhou_crosstalk_2025} & & crosstalk & $\Delta$threshold & \\
\onlinecite{maan_correlated_2026} & & corr.\ prior & decoder & \\
This work & routed & kernel model & Monte Carlo & exposure \\
\end{tabular}
\end{ruledtabular}
\end{table}

\begin{figure*}[!t]
\centering
\includegraphics[width=0.97\textwidth]{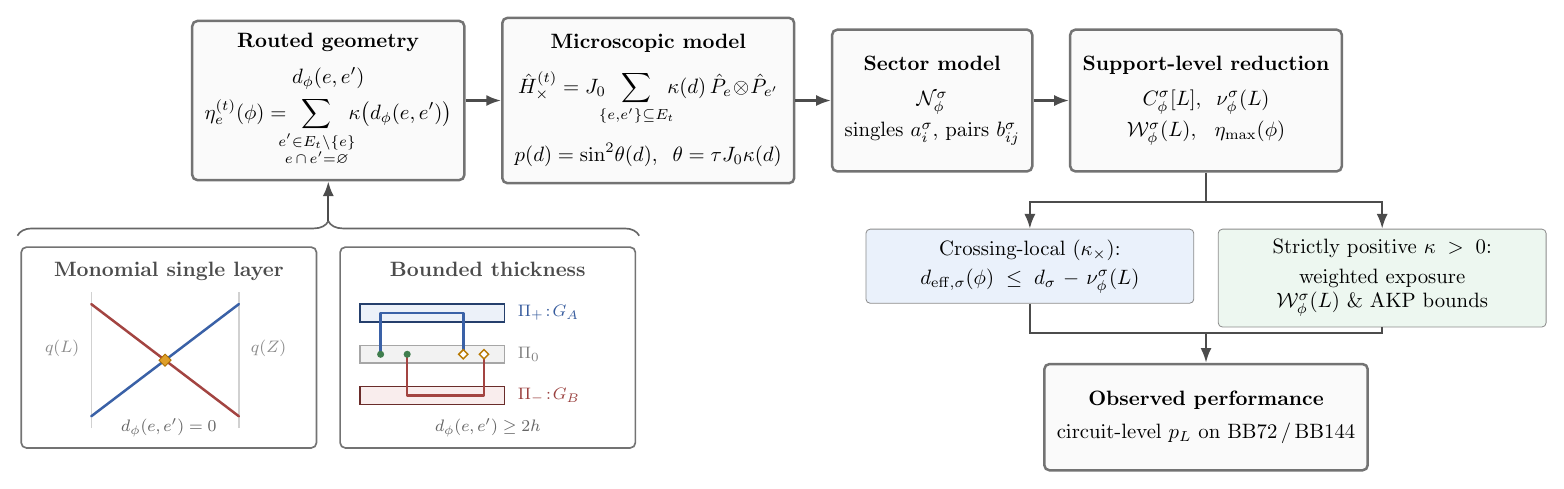}
\caption{Geometry-to-noise-to-performance pipeline. The left block contrasts the monomial single-plane layout with the biplanar bounded-thickness layout. Central boxes: microscopic Hamiltonian, twirl, retained channel, and support-graph objects. Right: circuit-level logical performance on BB72 and BB144. Branches beneath the support box distinguish the crossing-local diagnostic ($\kappa_\times$) from the strictly-positive-kernel regime ($\kappa>0$).}
\label{fig:pipeline}
\end{figure*}

The central question is whether routed geometry, on its own, changes the dominant correlated-fault structure and logical performance of a fixed code and syndrome-extraction schedule. Within the same-tick model studied below, it does. The analysis yields a combinatorial crossing diagnostic in the zero-separation limit and, under strictly positive distance-decay kernels, a weighted support-exposure metric that replaces the crossing count at finite coupling. The remainder of the paper derives these quantities and tests them at circuit level on BB72 and BB144.

We first formulate the same-tick model in the interaction-dominated regime and carry out the Pauli twirl on the resulting pair Hamiltonian. We then derive the retained data channel of single and pair faults that underlies the circuit-level sampling model used in the simulations.

Two geometry metrics arise from this channel. Under the assumption that every support vertex admits a retained single-location sector fault, the effective fault weight on a logical support equals the support size minus the matching number of the retained correlation graph induced on that support. The matching number is the size of a maximum set of retained pair-fault locations on the support with no shared qubits. For the worked BB72 weight-$6$ pure-$q(L)$ support---meaning a logical operator supported entirely on the $q(L)$ register---this effective weight drops from $6$ to $3$ in the monomial embedding and is unchanged in the biplanar embedding. Under any strictly positive kernel for which every support pair appears in at least one retained round with finite separation, the retained support graph saturates, the matching-number bound becomes embedding-independent, and weighted exposure---the total retained pair weight on the support---takes over as the finite-coupling discriminant between embeddings. For two-dimensional layouts with a regularized algebraic kernel of decay exponent $\alpha$, the same analysis yields the planar summability threshold $\alpha > 2$ and an AKP smallness criterion.

These results lead to a design objective for single-layer embeddings, where weighted exposure on the chosen logical family takes the place of routing depth. On the BB72 pure-$q(L)$ logical family, exhaustive enumeration of the \bbSeventyTwoPureLFamilyCount{} weight-$6$ supports reduces the maximum weighted exposure by \bbSeventyTwoLAImprovementPct\% relative to the monomial layout. BB72 is the main validation target, and BB144 provides a scaling check. Figure~\ref{fig:pipeline} summarizes the geometry-to-performance chain.

The paper's claims fall into four categories, tracked explicitly in Appendix~\ref{app:tables}. The microscopic-to-channel reduction, the support-graph identity for the effective fault weight, and the two-sided exposure bound are proved from Assumptions 1--3. The retention of only single and pair faults is controlled to fourth order in the per-tick coupling strength by the diamond-norm bound of Proposition~\ref{prop:truncation} and Theorem~\ref{thm:schedule-prop}. The retained channel form, the logical-aware objective, and the linearized single/pair sampling model are stated as modeling choices on which the finite-coupling analysis depends. The tracking of the logical error rate (LER) by weighted exposure, and the ordering of the monomial, biplanar, and logical-aware layouts, is empirical.

Sections~\ref{sec:model} and~\ref{sec:geom} introduce the microscopic model, the retained channel, the embedding families, and the kernel families. Section~\ref{sec:theory} extracts the two geometry metrics. Section~\ref{sec:logical-aware} turns the metrics into a single-layer design objective, and Section~\ref{sec:results} tests the resulting claims at circuit level. Appendices~\ref{app:two-block}--\ref{app:hardware} contain the general two-block decomposition, supplementary figures, diagnostics, and hardware-data comparisons.

\section{Microscopic model and retained effective channel}
\label{sec:model}

This section (i) introduces the open-system Hamiltonian and the proximity-kernel assumption that fixes the geometry-induced interaction, (ii) carries out the same-tick Pauli twirl on the resulting tick-level Hamiltonian, and (iii) reduces the twirled channel to the retained single-fault and pair-fault data channel used in the simulations.

We study repeated syndrome extraction for CSS codes~\cite{calderbank_shor_1996,steane_1996}, in which $X$-type and $Z$-type errors are corrected independently by separate parity-check matrices $H_X$ and $H_Z$ satisfying $H_X H_Z^\top=0$ over the binary field $\F$. A code encoding $k$ logical qubits into $n$ physical qubits with minimum logical-operator weight $d$ is denoted $\bb{n,k,d}$.

A BB code over the group algebra $\F[\mathbb{Z}_\ell\times\mathbb{Z}_m]$---formal $\F$-linear combinations of monomials $x^a y^b$ in the abelian group $\mathbb{Z}_\ell\times\mathbb{Z}_m$---has $n=2\ell m$ data qubits partitioned into four registers $q(X)$, $q(L)$, $q(R)$, and $q(Z)$ of size $M:=\ell m$. The two parity-check matrices are circulants $A$ and $B$, each determined by three monomial terms, and satisfy $AB^\top=BA^\top$~\cite{bravyi_memory_2024}.

The code-independent constructions below apply to any routed CSS schedule. The numerical study uses the depth-$8$ BB schedule of Ref.~\onlinecite{bravyi_memory_2024}, in which one syndrome cycle comprises three $A$-family and three $B$-family \textsc{cnot} rounds together with initialization and measurement. Of the four registers, $q(L)$ and $q(R)$ carry data qubits, $q(X)$ and $q(Z)$ carry syndrome ancillas, and all four have size $M=\ell m$. The sector-relevant rounds are the three $B$ rounds determined by the monomial terms of $B=B_1+B_2+B_3$. In the $X$ sector, these are the $q(L)\to q(Z)$ \textsc{cnot}s, and a same-tick geometry event in one of them lands, after Clifford propagation and ancilla elimination, on two $q(L)$ data qubits. In the $Z$ sector, they are the $q(X)\to q(R)$ \textsc{cnot}s, and the corresponding pair event lands on two $q(R)$ data qubits.

\subsection{Open system and routed geometry}

Throughout this section, we set $\hbar=1$. The microscopic starting point is the open-system Hamiltonian decomposition
\begin{equation}
\hat H(t)=\hat H_{\mathrm S}(t)+\hat H_{\mathrm B}+\hat H_{\mathrm{SB}}(t)+\hat\varepsilon(t),
\label{eq:open-system-main}
\end{equation}
where $\hat H_{\mathrm S}(t)$ implements the ideal syndrome-extraction schedule, $\hat H_{\mathrm B}$ is the bath Hamiltonian, $\hat H_{\mathrm{SB}}(t)$ is the system--bath coupling, and $\hat\varepsilon(t)$ collects coherent system-only imperfections. The geometry-induced coherent contribution is denoted by $\Hx(t)$ and treated as part of $\hat\varepsilon(t)$, so routed separation controls a distinct part of the coherent error budget.

An embedding $\phi$ fixes qubit positions and a routed curve for each active two-qubit gate block. All in-plane distances are measured in units of the placement pitch, defined as the center-to-center spacing between adjacent qubit sites on the chip. A bounded-thickness layout uses finitely many parallel routing planes separated by a fixed layer spacing, and same-tick pair couplings depend only on the routed closest approach.

\begin{assumption}[Tick structure and routed separations]
A tick is a time window of duration $\tau$ during which a set $E_t$ of disjoint two-qubit gates is executed, together with any accompanying single-qubit gates, idles, measurements, or resets. Each element $e\in E_t$ is an unordered pair of qubit indices, denoting the two qubits acted on by the gate; disjointness means $e\cap e'=\varnothing$ for distinct $e,e'\in E_t$. For a routed embedding $\phi$, the separation between simultaneously active gate blocks is the finite geometric quantity $d_{\phi}(e,e')\ge 0$.
\end{assumption}

Any Hermitian perturbation acting on two disjoint gate blocks has a unique decomposition into a scalar, two single-block operators, and an inter-block interaction with vanishing partial traces (Appendix~\ref{app:two-block}). Under a Pauli twirl, only the inter-block interaction generates correlated pair faults at leading order in the per-pair coupling phase, while the single-block components contribute at strictly higher order and fold into the independent single-block baseline. Residual inter-block couplings on superconducting-qubit platforms are predominantly of this interaction (ZZ-type) form~\cite{pettersson_zz_2024,aguila_flux_2025,zhou_crosstalk_2025}. The model below therefore keeps only the inter-block interaction and treats the single-block components as part of the independent single-block noise.

\begin{assumption}[Proximity kernel and inter-block interaction]
The geometry-induced perturbation between simultaneously active disjoint blocks $e,e'\in E_t$ is modeled by the interaction component of the general two-block decomposition (Appendix~\ref{app:two-block}). A dimensionless, nonnegative, monotonically nonincreasing proximity kernel $\kappa$ with $\kappa(0)=1$ and a coupling scale $J_0>0$ with units of inverse time specify the microscopic Hamiltonian
\begin{equation}
\Hx(e,e')=J_0\,\kappa\left(d_{\phi}(e,e')\right)\,\hat P_e\otimes\hat P_{e'},
\label{eq:pair-ham-main}
\end{equation}
where $\hat P_e$ and $\hat P_{e'}$ are traceless Hermitian involutions ($\hat P^2=\I$, $\hat P^{\dagger}=\hat P$) on the two active gate blocks, each determined by the gate type and block role in the schedule. The tensor product $\hat P_e\otimes\hat P_{e'}$ is itself a traceless Hermitian involution on the joint block space.
\end{assumption}

In the implemented BB schedule, each \textsc{cnot} acts on one data qubit and one ancilla qubit. The block Pauli $\hat P_e$ on a \textsc{cnot} block is the factor of the two-qubit Pauli that propagates, through the remaining Clifford gates and ancilla elimination, to a nontrivial data fault in the chosen CSS sector. For the $X$ sector, this is $\hat P_e = X_{\mathrm{data}}\otimes I_{\mathrm{anc}}$ on a $q(L)\!\to\!q(Z)$ \textsc{cnot} block, which propagates to an $X$ fault on the $q(L)$ data qubit; for the $Z$ sector, it is $\hat P_e = Z_{\mathrm{anc}}\otimes Z_{\mathrm{data}}$ on a $q(X)\!\to\!q(R)$ block, which propagates to a $Z$ fault on the $q(R)$ data qubit. These are the two choices used in Theorem~\ref{thm:schedule-prop}.

Routed separation sets the pair-coupling strength $J_0\kappa(d_\phi(e,e'))$, while the kernel profile $\kappa$ is left as a free parameter family since the microscopic coupling depends on the architecture (layout, couplers, shielding, and calibration).

Summing Eq.~\eqref{eq:pair-ham-main} over all unordered simultaneously active block pairs in a tick gives the tickwise geometry Hamiltonian
\begin{equation}
\Hxt
:=
J_0
\sum_{\{e,e'\}\,\subseteq \,E_t}
\kappa\left(d_{\phi}(e,e')\right)\,
\hat P_e\otimes\hat P_{e'}.
\label{eq:tick-ham-main}
\end{equation}
Let $A_t := \{\{e,e'\}\mid e,e'\in E_t,\,e\neq e'\}$ denote the set of unordered simultaneously active gate-block pairs in tick $t$, and for each pair $a=\{e,e'\}\in A_t$ set $\hat S_a := \hat P_e\otimes\hat P_{e'}$. The per-pair phase accumulated over the tick duration $\tau$ is
\begin{equation}
\theta_a := \tau J_0\,\kappa\!\left(d_\phi(e,e')\right),
\label{eq:theta-a}
\end{equation}
and the tickwise coupling norm---an $\ell^2$ aggregate over same-tick pairs---is
\begin{equation}
\Theta_t^2 := \sum_{a\,\in A_t}\theta_a^2.
\label{eq:Theta-t}
\end{equation}
As shown below, $\Theta_t$ controls the rate at which multiple pair events coincide within the same tick. Because the gate blocks in $E_t$ are disjoint, the same-tick generators commute, $\left[\hat S_a,\hat S_b\right]=0$ for all $a,b\in A_t$, and the same-tick unitary factorizes as
\begin{equation}
\hat U_t = \prod_{a\,\in A_t} e^{-i\theta_a \hat S_a} = \prod_{a\,\in A_t}\!\left(\cos\theta_a\,\I - i\sin\theta_a\,\hat S_a\right).
\label{eq:tick-factorize}
\end{equation}
We write $\mathcal U_t(\hat\rho) := \hat U_t\,\hat\rho\,\hat U_t^{\dagger}$ for the channel it induces.
Each generator squares to the identity ($\hat S_a^2=\I$). Expanding the product in Eq.~\eqref{eq:tick-factorize} over subsets $B\subseteq A_t$ gives the parity expansion
\begin{equation}
\hat U_t = \sum_{B\,\subseteq\, A_t} \alpha_B(\theta)\,\hat S_B,
\label{eq:parity-expansion}
\end{equation}
with Pauli monomial $\hat S_B := \prod_{a\in B}\hat S_a$ and coefficients
\begin{equation}
\alpha_\varnothing := \prod_{a\,\in A_t}\cos\theta_a,
\label{eq:alpha-empty}
\end{equation}
and, for $B\neq\varnothing$,
\begin{equation}
\alpha_B := (-i)^{|B|}\prod_{a\,\in\, B}\sin\theta_a\prod_{a\,\notin\, B}\cos\theta_a.
\label{eq:alpha-B}
\end{equation}
Each $\hat S_B$ is a product of commuting involutions and is therefore a Pauli monomial on the active blocks. Two distinct subsets $B\neq B'$ may yield the same monomial $\hat S_B=\hat S_{B'}$ if block-Pauli factors cancel. For a Pauli $\hat Q\in\Pgroup_k$ on the $k$ qubits spanned by the active blocks, let $\mathcal Q(\hat\rho):=\hat Q\hat\rho\,\hat Q$ be the corresponding Pauli conjugation channel. The Pauli twirl $\mathcal T$ averages any channel $\mathcal C$ over the Pauli group by conjugation,
\begin{equation}
\mathcal T[\mathcal C]:=\frac{1}{4^k}\sum_{\hat Q\,\in\,\Pgroup_k}\mathcal Q\circ\mathcal C\circ\mathcal Q,
\label{eq:pauli-twirl-def}
\end{equation}
with $\Pgroup_k=\{I,X,Y,Z\}^{\otimes k}$ the $k$-qubit Pauli group modulo phases. Square brackets distinguish the twirl acting on a channel from the resulting channel acting on a state. Applied to the full-tick unitary, $\mathcal T$ groups the monomial collisions above as
\begin{equation}
\mathcal T[\mathcal U_t](\hat\rho) = \sum_{\hat Q}\left|\sum_{\substack{B\,\subseteq\, A_t\\\hat S_B\,=\,\hat Q}} \alpha_B(\theta)\right|^2\hat Q\hat\rho\,\hat Q,
\label{eq:tick-twirl}
\end{equation}
where the outer sum runs over distinct Pauli monomials $\hat Q$ on the active blocks. At leading order, each single-pair term ($|B|=1$) produces a distinct monomial $\hat S_{\{a\}}$, and the corresponding single-pair probability is
\begin{align}
\left|\alpha_{\{a\}}\right|^2 &= \sin^2\theta_a\prod_{b\,\neq\, a}\cos^2\theta_b \\
&= \sin^2\theta_a + O(\Theta_t^4).
\end{align}
All multi-pair terms and cross-subset interference contribute at $O(\Theta_t^4)$. The retained model is the leading-order restriction of this expansion.

The total same-tick \emph{block exposure} of gate block $e$ is
\begin{equation}
\eta_e^{(t)}(\phi)  := \sum_{\substack{e'\in \,E_t \\ e'\neq \,e}}\kappa\left(d_{\phi}(e,e')\right).
\label{eq:eta-main}
\end{equation}
We use block exposure for AKP comparisons and layout audits; the distinct \emph{support exposure} $\Ew_\phi^\sigma(L)$, defined in Sec.~\ref{sec:theory}, aggregates retained pair weights over a logical support.

\subsection{Pauli twirl and two-block corollary}

Where convenient, we write the per-pair phase as a function of routed separation, $\theta(d):=\tau J_0\kappa(d)$, so that $\theta_a=\theta(d_\phi(e,e'))$ for $a=\{e,e'\}\in A_t$ [cf.\ Eq.~\eqref{eq:theta-a}]. The pair unitary on a single simultaneously active pair $e,e'\in E_t$ is then
\begin{equation}
\hat U_{e,e'}=\exp\left[-i\theta(d)\,\hat P_e\otimes\hat P_{e'}\right].
\label{eq:block-unitary}
\end{equation}
Let $\mathcal U_{e,e'}(\hat\rho):=\hat U_{e,e'}\hat\rho\,\hat U_{e,e'}^{\dagger}$ denote the corresponding unitary channel, and recall the Pauli twirl $\mathcal T$ of Eq.~\eqref{eq:pauli-twirl-def}.

\begin{theorem}[General single-block Pauli twirl]\footnote{Under randomized compiling~\cite{wallman_emerson_2016}, the averaged channel over Pauli frames equals $\mathcal T[\mathcal C]$, so the twirl describes the realized average channel rather than an approximation.}
\label{thm:general-twirl}
Let $\hat U=e^{-i\theta \hat P}$ on $k$ qubits, with $\hat P$ a Hermitian involution ($\hat P^2=\I$, $\hat P^{\dagger}=\hat P$), and write $\mathcal U(\hat\rho):=\hat U\hat\rho\,\hat U^\dagger$ for the unitary channel it induces. Then
\begin{equation}
\mathcal T[\mathcal U](\hat\rho)=\sum_{\hat R\,\in\,\Pgroup_k}p_{\hat R}\,\hat R\hat\rho\hat R
\label{eq:general-twirl}
\end{equation}
where
\begin{equation}
    p_{\hat R}=\frac{1}{4^k}\left|\operatorname{tr}\hat R\hat U\right|^2.
\end{equation}
\end{theorem}

\begin{proof}
Expanding $\hat U=\cos\theta\,\I-i\sin\theta\,\hat P$ and taking the trace against $\hat R$ gives
\begin{equation}
\operatorname{tr}\hat R\hat U=\cos\theta\,\operatorname{tr}\hat R-i\sin\theta\,\operatorname{tr}\hat R\hat P.
\end{equation}
For $\hat R=\I$, $\operatorname{tr}\hat R=2^k$, so
\begin{equation}
p_{\I}=\cos^2\theta+\frac{\sin^2\theta}{4^k}\left|\operatorname{tr}\hat P\right|^2.
\label{eq:pI}
\end{equation}
For $\hat R\neq\I$, $\operatorname{tr}\hat R=0$, so
\begin{equation}
p_{\hat R}=\frac{\sin^2\theta}{4^k}\left|\operatorname{tr}\hat R\hat P\right|^2.
\label{eq:pR}
\end{equation}
If in addition $\hat P$ is traceless, $p_{\I}=\cos^2\theta$ and
\begin{equation}
\sum_{\hat R\,\neq\,\I}p_{\hat R}=\sin^2\theta
\end{equation}
by normalization.
\end{proof}

Since $\hat P_e\otimes\hat P_{e'}$ is a traceless Hermitian involution on the joint block space (Assumption~2), Theorem~\ref{thm:general-twirl} applies directly. Define the \emph{twirled pair-fault probability}
\begin{equation}
p(d):=\sin^2\theta(d),
\label{eq:p-def}
\end{equation}
i.e., the probability that the correlated $\hat P_e\otimes\hat P_{e'}$ fault occurs on a simultaneously active pair at routed separation $d$.

\begin{corollary}[Interaction two-block twirl]
\label{cor:twirl}
If $\hat P_e$ and $\hat P_{e'}$ are traceless Pauli operators on disjoint blocks, then
\begin{equation}
\mathcal T[\mathcal U_{e,e'}]=\bigl(1-p(d)\bigr)\,\mathrm{Id}+p(d)\,\mathcal P_{e,e'},
\label{eq:twirled-block}
\end{equation}
where $\mathcal P_{e,e'}(\hat\rho)=(\hat P_e\otimes\hat P_{e'})\,\hat\rho\,(\hat P_e\otimes\hat P_{e'})$ is the correlated pair-fault channel.
\end{corollary}

\begin{proof}
$\hat P_e\otimes\hat P_{e'}$ is traceless and satisfies $(\hat P_e\otimes\hat P_{e'})^2=\I$, so Theorem~\ref{thm:general-twirl} gives $p_{\I}=\cos^2\theta$ and $p_{\hat P_e\otimes\hat P_{e'}}=\sin^2\theta$. All other Pauli weights vanish by orthogonality.
\end{proof}

Equation~\eqref{eq:twirled-block} is the leading geometry-induced effect at the block level. Single-block faults arise only from the local baseline noise, rather than the inter-block coupling. The data-level model used in the simulations is obtained from this block-level channel below by propagating it through the syndrome cycle, eliminating the ancilla qubits, and restricting to the chosen sector.

Three levels of geometry-dependent quantity appear in the derivation that follows:
\begin{enumerate}[label=(\roman*),leftmargin=2em]
\item the microscopic Hamiltonian amplitude $J_0\kappa(d)$, or equivalently the dimensionless phase $\theta(d)$ [Eq.~\eqref{eq:theta-a}];
\item the twirled pair-fault probability $p(d)=\sin^2\theta(d)$ [Eq.~\eqref{eq:p-def}];
\item the retained sector coefficient $q_\kappa(d)$ [Eq.~\eqref{eq:q-kappa-def} below].
\end{enumerate}
The first is a property of the Hamiltonian before twirling. The second is the twirled pair-fault probability at the block level. The third absorbs the full Clifford propagation of a single pair event through one syndrome-extraction round, ancilla elimination, and sector restriction. Because the extraction circuit is Clifford and the twirled channel is diagonal in the Pauli basis, $q_\kappa(d)$ is a finite computation evaluated numerically for each kernel; it is fixed by the schedule and the microscopic phases, and it is not a closed-form function of $p(d)$ alone. The AKP summability statements below concern the microscopic amplitudes in (i) before twirling, whereas the stabilizer simulations use the retained coefficients in (iii) after propagation and truncation. In the production dataset, we set $\tau=1$---and, for the regularized algebraic kernel, $r_0=1$---so the plotted coupling $J_0\tau$ is dimensionless.

\subsection{Retained sector model}

Fix a CSS sector $\sigma\in\{X,Z\}$: the $X$ sector concerns $X$-type data faults detected by $Z$-check measurements, and the $Z$ sector concerns the converse. After propagation through one syndrome-extraction round, elimination of the ancilla qubits (by tracing over their measurement outcomes), and restriction to the chosen sector, we retain the nonzero single-data and pair-data contributions and write the effective model as
\begin{widetext}
\begin{equation}
\mathcal N^{\sigma}_{\phi,\le 2}=\left(1-\sum_i a_i^{\sigma}-\sum_{i\,<\,j} b_{ij}^{\sigma}\right)\mathrm{Id}+\sum_i a_i^{\sigma}\,\mathcal P_i^{\sigma}+\sum_{i\,<\,j} b_{ij}^{\sigma}\,\mathcal P_{ij}^{\sigma}.
\label{eq:retained-channel}
\end{equation}
\end{widetext}
Here $\mathrm{Id}(\hat\rho):=\hat\rho$ is the identity channel, $\mathcal P_i^{\sigma}$ applies the sector-relevant single-qubit Pauli on data location $i$, and $\mathcal P_{ij}^{\sigma}$ applies the corresponding pair operator on data locations $i$ and $j$. All coefficients $a_i^{\sigma}\ge 0$, $b_{ij}^{\sigma}\ge 0$ are nonnegative retained Pauli probabilities summing to at most $1$. Equation~\eqref{eq:retained-channel} defines the full retained \emph{geometry-induced} channel after one-round propagation, ancilla elimination, and sector restriction; the complete extraction map is obtained by composing this channel with the baseline local circuit noise of Sec.~\ref{sec:syndrome}. The $b_{ij}^{\sigma}$ are the pair weights used directly as graph edges in Sec.~\ref{sec:theory}.

Write $q_\kappa(d)$ for the \emph{retained sector coefficient} carried by a single sector-relevant same-round pair event at routed separation $d$; that is, the contribution of one such event to $b_{ij}^{\sigma}$ under Corollary~\ref{cor:linearized} below. This is the geometry-dependent factor applied to every sector-relevant same-round pair contribution. At leading order in the tickwise coupling norm,
\begin{equation}
q_\kappa(d) = \sin^2\!\theta(d) + O(\Theta_t^4) = p(d) + O(\Theta_t^4),
\label{eq:q-kappa-def}
\end{equation}
so $q_\kappa$ and $p$ agree to $O(\Theta_t^2)$ and differ only through multi-pair corrections.

The truncation to single and pair terms in Eq.~\eqref{eq:retained-channel} is controlled by the following bound on the discarded higher-weight Pauli mass.

\begin{proposition}[Controlled low-body truncation]
\label{prop:truncation}
Let
\begin{equation}
\mathcal N_\phi^\sigma := \sum_P p_P\,\mathcal U_P
\end{equation}
be the full Pauli-twirled data channel after circuit propagation, ancilla elimination, and sector restriction, with $\mathcal U_P(\hat\rho)=\hat P\hat\rho\hat P$ and $P$ ranging over sector-relevant data Paulis. Let $\mathcal N_{\phi,\le 2}^\sigma$ be the retained channel of Eq.~\eqref{eq:retained-channel}, obtained from $\mathcal N_\phi^\sigma$ by absorbing all weight-${\ge 3}$ Pauli mass into the identity term, and let
\begin{equation}
    M_{\ge 3}:=\sum_{\wt P\,\ge\, 3}p_P.
\end{equation}
Then, in the diamond norm $\lVert\,\cdot\,\rVert_\diamond$,
\begin{equation}
\bigl\lVert \mathcal N_\phi^\sigma - \mathcal N_{\phi,\le 2}^\sigma \bigr\rVert_\diamond\, \le\, 2\,M_{\ge 3}.
\label{eq:truncation-bound}
\end{equation}
\end{proposition}

\begin{proof}
The difference channel is
\begin{equation*}
\sum_{\wt P\,\ge\, 3}p_P\left(\mathcal U_P-\mathrm{Id}\right).
\end{equation*}
By the triangle inequality,
\begin{equation}
\bigl\lVert \mathcal N_\phi^\sigma - \mathcal N_{\phi,\le 2}^\sigma \bigr\rVert_\diamond
\,\le\sum_{\wt P\,\ge\, 3} p_P\,\lVert \mathcal U_P - \mathrm{Id}\rVert_\diamond.
\end{equation}
Each $\mathcal U_P$ and $\mathrm{Id}$ is a quantum channel, so $\lVert \mathcal U_P-\mathrm{Id}\rVert_\diamond \le 2$, which gives the bound.
\end{proof}

Since each single-pair event in the BB schedule propagates to a weight-$2$ data Pauli, the discarded mass $M_{\ge 3}$ is controlled by multi-pair coincidences and can be bounded more tightly.

\begin{theorem}[Schedule-level weight-$2$ propagation]
\label{thm:schedule-prop}
In the implemented depth-$8$ BB extraction schedule, every single-pair geometry event $\hat P_e\otimes\hat P_{e'}$ in a sector-relevant $B$ round propagates through the remaining Clifford gates and ancilla elimination to a weight-$2$ data Pauli. Consequently, the discarded weight-$\ge 3$ Pauli mass satisfies
\begin{equation}
M_{\ge 3}^{\sigma}(\phi) = O(\Theta_t^4) \qquad (\Theta_t\to 0).
\label{eq:schedule-remainder}
\end{equation}
\end{theorem}

\begin{proof}
By Corollary~\ref{cor:twirl}, a single-pair event in a $B$ round applies the sector-relevant Pauli $\hat P_e\otimes\hat P_{e'}$ to the two data qubits of the interacting gate blocks. In the $X$ sector, $\hat P_e$ acts as $X$ on a $q(L)$ data qubit (the control of the $q(L)\to q(Z)$ \textsc{cnot}). Clifford conjugation by subsequent \textsc{cnot}s in the depth-$8$ schedule maps $X$ on a \textsc{cnot} control to $X$ on both control and target; the target is an ancilla qubit ($q(Z)$) that is measured and eliminated. The data-level weight therefore remains $2$. The $Z$-sector argument is analogous, with the roles of $q(L)$ and $q(R)$ exchanged.

Since every single-pair event propagates to weight $\le 2$, weight-$\ge 3$ contributions require two or more pair events in the same tick. Set $x_a:=\sin^2\theta_a$, let $Y_a\sim\mathrm{Bernoulli}(x_a)$ be independent, and write $N:=\sum_a Y_a$. From the parity expansion of Eq.~\eqref{eq:tick-twirl},
\begin{equation}
|\alpha_B|^2=\prod_{a\,\in\, B}x_a\prod_{a\,\notin\, B}(1-x_a).
\end{equation}
Using $\mathbf 1_{N\,\ge\, 2}\,\le\,\binom{N}{2}$ and $\sum_a x_a\le\sum_a\theta_a^2=\Theta_t^2$,
\begin{align}
\Pr[N\ge 2]\,&\le\,\mathbb E\!\left[\tbinom{N}{2}\right]\\[0.33em]
&=\, \sum_{a\,<\,b}x_a x_b\\[0.3em]
&\le\, \frac{1}{2}\Theta_t^4.
\end{align}
The full-tick twirl in Eq.~\eqref{eq:tick-twirl} groups subsets by monomial collision, so the squared-amplitude sum over distinct monomials picks up a Cauchy--Schwarz multiplicity factor
\begin{equation}
m_t:=\max_{\hat Q}\bigl|\{B\subseteq A_t\mid \hat S_B=\hat Q,\,|B|\ge 2\}\bigr|.
\end{equation}
The depth-$8$ BB schedule has finitely many tick types and each tick contains finitely many active blocks, so $m_\ast:=\sup_t m_t<\infty$. This finite multiplicity absorbs into the $O(\Theta_t^4)$ prefactor of Eq.~\eqref{eq:schedule-remainder}.
\end{proof}

\begin{remark}
For any finite coupling window $\Theta_t\le\Theta_\ast$, Eq.~\eqref{eq:schedule-remainder} holds with an explicit constant $C_\sigma(\Theta_\ast)$ absorbing the proof's factor $\frac{1}{2}\exp\Theta_\ast^2$ and the schedule-dependent multiplicity $m_\ast$. This constant depends on the specific BB schedule instance and is not claimed to be uniform across code sizes.
\end{remark}

\begin{figure*}[!tbp]
\centering
\includegraphics[width=0.94\textwidth]{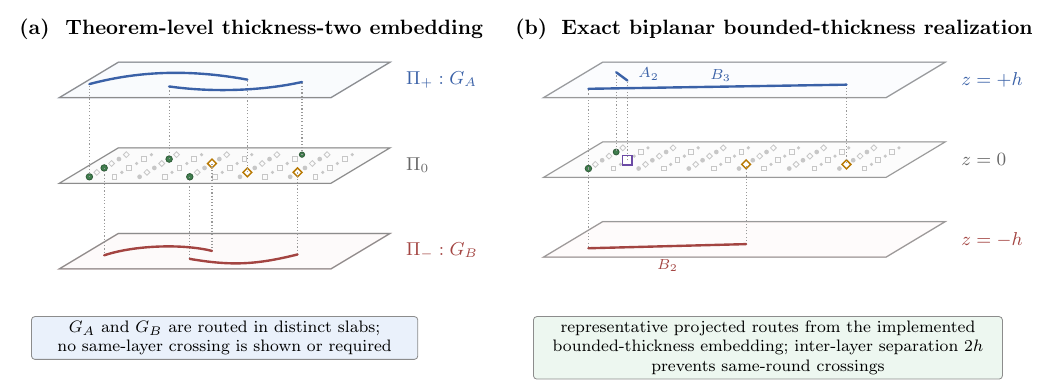}
\caption{Theorem-level and implemented bounded-thickness embeddings. (a) Schematic thickness-two cartoon. All qubit sites lie in the base plane $\Pi_0$ (shown as the toric register grid), while the relevant same-round route families $G_A$ and $G_B$ are assigned to distinct routing slabs $\Pi_+$ and $\Pi_-$. Dotted lines indicate layer access and smooth curves show schematic non-crossing in-plane routes. No same-layer crossing is shown or required. (b) Biplanar bounded-thickness realization used in the numerical study. Solid lines show representative in-plane traverses, and dotted lines show the layer access paths connecting $\Pi_0$ to the routing slabs. The two blue traverses in $\Pi_+$ ($B_3$ and $A_2$, both in $G_A$) connect qubits at different toric-grid depths, so they project to different positions within the slab. The inter-layer separation $2h$ prevents same-round crossings across $G_A$ and $G_B$. This panel visualizes the realized bounded-thickness routing and does not reconstruct the hardware geometry of Ref.~\onlinecite{bravyi_memory_2024}.}
\label{fig:biplanar-vs-implementation}
\end{figure*}

\begin{corollary}[Linearized retained coefficients]
\label{cor:linearized}
Under the conditions of Theorem~\ref{thm:schedule-prop}, the retained pair coefficients satisfy
\begin{equation}
b_{ij}^{\sigma}(\phi) = \sum_{a\,\in A_t} M_{ij,a}^{\sigma}\,\sin^2\theta_a + O\!\left(\Theta_t^4\right),
\end{equation}
where $M_{ij,a}^{\sigma}\in\{0,1\}$ is the schedule-determined indicator that the single-pair event $a$ propagates to the data-pair $(i,j)$ in sector $\sigma$.
\end{corollary}

For the full depth-$8$ cycle with $R$ relevant rounds, the roundwise geometry increments compose as
\begin{widetext}
\begin{equation}
\Delta\mathcal{N}_{\mathrm{cycle}}^{\sigma} = \sum_r \Delta\mathcal{N}_{r,\le 2}^{\sigma}
+ O\!\left(\sum_r \Theta_r^4 + \sum_{r\,<\,s}\Theta_r^2\Theta_s^2\right),
\label{eq:cycle-compose}
\end{equation}
\end{widetext}
so the full-cycle retained model inherits the roundwise $O(\Theta^4)$ control. From this point onward, the geometry enters the finite-code analysis only through the retained coefficients $a_i^\sigma$ and $b_{ij}^\sigma$, and through the weighted correlation graph they induce.

The baseline data channel at $\theta=0$ already contains higher-weight Pauli components from the stabilizer projection, handled by the decoder. The geometry-induced \emph{incremental} channel
\begin{equation}
\Delta\mathcal N_\phi^\sigma(\theta):=\mathcal N_\phi^\sigma(\theta)-\mathcal N_\phi^\sigma(0)
\end{equation}
need not introduce new higher-weight mass; Appendix~\ref{app:microscopic} (Fig.~\ref{fig:weight-audit}) verifies on a minimal one-round subcircuit that $\Delta M_{\ge 3}(\theta)<0$, so the geometry correction generates only weight-$1$ and weight-$2$ mass in that setting.

\section{Embeddings, kernels, and the BB72 reference support}
\label{sec:geom}

The retained channel of Eq.~\eqref{eq:retained-channel} depends on a routed embedding $\phi$ (Assumption~1) and on a proximity kernel $\kappa$ (Assumption~2). This section fixes the two embedding families, the three kernel families, and the worked BB72 support used throughout the rest of the paper.

\begin{figure*}[!tbp]
\centering
\includegraphics[width=0.92\textwidth]{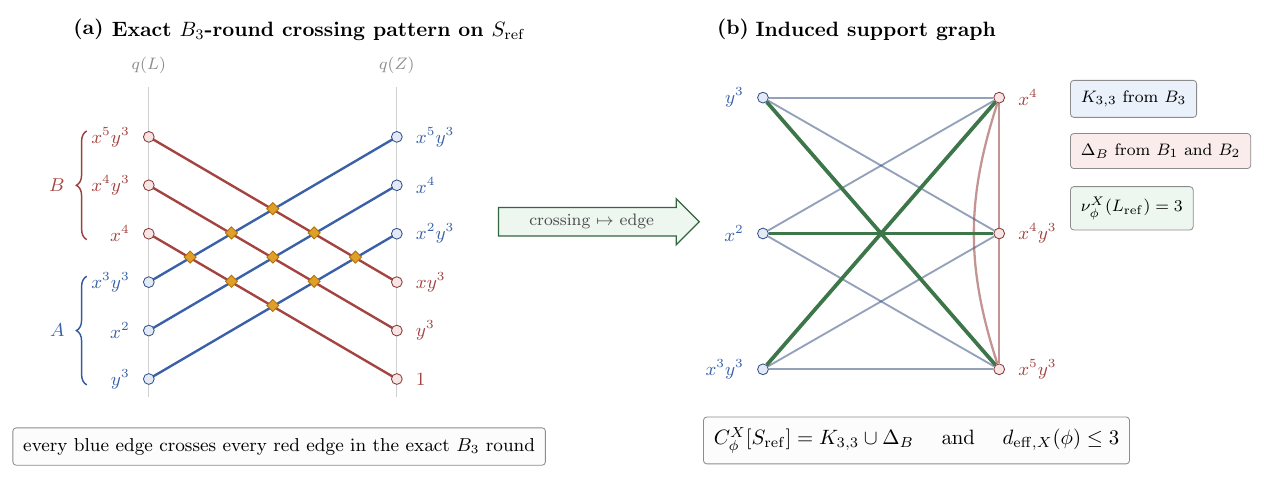}
\caption{BB72 reference support $S_{\mathrm{ref}}$ and the passage from crossings to the retained support graph. (a) $B_3$-round crossing picture on $S_{\mathrm{ref}}$ in the four-column monomial embedding; each same-round blue--red crossing pair induces an edge of the support graph. (b) Support graph after collecting the $B_3$ contributions and adding the $\Delta_B$ triangle from rounds $B_1$ and $B_2$. The highlighted matching has size $\nu_{\phi}^{X}(L_{\mathrm{ref}})=3$.}
\label{fig:bridge}
\end{figure*}

\subsection{Embedding families}
\label{sec:embeddings}

We compare two embedding families throughout. The first is a deterministic four-column monomial layout in which the $q(X)$, $q(L)$, $q(R)$, and $q(Z)$ registers occupy parallel columns and same-round gates are routed as straight segments. The second is a \emph{biplanar bounded-thickness} layout that implements the layer partition $G_A=\{A_2,A_3,B_3\}$ and $G_B=\{A_1,B_1,B_2\}$ from Ref.~\onlinecite{bravyi_memory_2024} and routes the relevant same-round edges without crossings within each layer. Figure~\ref{fig:biplanar-vs-implementation} contrasts the slab picture used in the theorem with the routed realization used in the numerics. The layer split and the no-same-round-crossing property are imported from Ref.~\onlinecite{bravyi_memory_2024} without reconstructing the full hardware geometry.

In the biplanar implementation, each routed edge receives an infinitesimal lane offset $\varepsilon_i\propto h/n$ within its layer, so that the 3D polylines are distinct. These offsets regularize the geometry engine and vanish in the large-code limit. The bounded-thickness separation between the two layers is the inter-layer distance $2h$. Under any strictly positive kernel, same-layer pairs at separation $\varepsilon_i\to 0$ still contribute $\kappa(0)$, so replacing the infinitesimal offsets with any fixed within-layer clearance changes per-pair kernel values but preserves the embedding hierarchy.

\subsection{Proximity kernels}
\label{sec:kernels}

Three kernel families appear in the paper. The \emph{crossing kernel} retains only zero-separation pairs,
\begin{equation}
\kappa_{\times}(d)=\mathbf 1_{d\,=\,0},
\label{eq:crossing-kernel}
\end{equation}
and isolates the combinatorial same-round crossing mechanism. The \emph{regularized algebraic kernel}
\begin{equation}
\kappa_{\alpha,r_0}(d)=\left(1+\frac{d}{r_0}\right)^{-\alpha}
\label{eq:powerlaw}
\end{equation}
has decay exponent $\alpha>0$ and regularization length $r_0>0$, and is the canonical two-dimensional distance-decay profile used to describe superconducting-qubit crosstalk~\cite{barrett_flux_2023}. The \emph{exponential kernel}
\begin{equation}
\kappa_{\xi}(d)=e^{-d/\xi}
\label{eq:exp-kernel}
\end{equation}
has decay length $\xi>0$ and describes screened or short-range residual couplings~\cite{kosen_crosstalk_2024}. Under $\kappa_\times$ the retained correlation graph is purely combinatorial. Under any strictly positive kernel, every finite same-round separation contributes a nonzero pair weight, and weighted exposure replaces crossing count as the discriminator between embeddings (Sec.~\ref{sec:theory}).

\subsection{Crossing criterion for four-column embeddings}
\label{sec:crossing-criterion}

In a straight-line four-column embedding, let $\sigma_L,\sigma_Z\in S_M$ be the permutations that assign vertical positions to the $q(L)$ and $q(Z)$ rows, with $S_M$ the symmetric group on $M=\ell m$ indices. The relevant same-round $q(L)\to q(Z)$ edge between data row $\alpha$ and its round-$r$ target $B_r\alpha$ (the image of $\alpha$ under the monomial action of the $r$-th term of $B$) crosses the one between $\beta$ and $B_r\beta$ if and only if
\begin{equation}
\left(\sigma_L(\alpha)-\sigma_L(\beta)\right)
\left(\sigma_Z(B_r\alpha)-\sigma_Z(B_r\beta)\right)<0.
\label{eq:cross-criterion}
\end{equation}
This is the order-reversal criterion for segments joining two parallel lines. The biplanar embedding has no such crossings by construction, because the relevant $B_1$, $B_2$, and $B_3$ rounds are routed in distinct layers and the same-round routes are planar within each layer.

\subsection{BB72 reference support}
\label{sec:ref-support}

BB72 is the smallest published BB memory benchmark in Ref.~\onlinecite{bravyi_memory_2024}. Throughout the rest of the paper, the pure-$q(L)$ support
\begin{equation}
S_{\mathrm{ref}}=\{3,12,21,24,27,33\}
\label{eq:ref-support}
\end{equation}
serves as a concrete minimum-weight $X$-sector logical support on BB72, on which all subsequent worked calculations are evaluated. We write $L_{\mathrm{ref}}$ for the corresponding $X$-sector logical operator, so that $\supp L_{\mathrm{ref}} =S_{\mathrm{ref}}$. Figure~\ref{fig:bridge} shows the passage from the monomial-embedding crossing pattern on $S_{\mathrm{ref}}$ to the retained support graph. The matching-number and weighted-exposure evaluations on $S_{\mathrm{ref}}$ are in Sec.~\ref{sec:theory}. Appendix~\ref{app:geometry} gives the toric-base placement.

\section{Effective distance, weighted exposure, and the AKP-type compatibility}
\label{sec:theory}

To turn the retained pair structure of Eq.~\eqref{eq:retained-channel} into a finite-code metric, we encode the nonzero retained pair coefficients as edges of a weighted graph on the sector data qubits. The two geometry metrics used in the rest of the paper read off directly from this graph. For a fixed embedding $\phi$ and sector $\sigma$, the \emph{weighted correlation graph} is
\begin{equation}
C_{\phi}^{\sigma}=(V,E,w),
\end{equation}
whose vertices $V$ are the sector data qubits, whose edges $E$ are the retained pair locations with nonzero coefficient in Eq.~\eqref{eq:retained-channel}, and whose edge weights $w_{ij}:=b_{ij}^{\sigma}$ are the retained pair probabilities. A \emph{logical operator} is a Pauli operator that commutes with every stabilizer but is not itself a stabilizer; its \emph{support} is $S=\supp L$, the set of qubits on which it acts nontrivially. Given a sector logical operator $L$, write $C_{\phi}^{\sigma}[S]$ for the subgraph of $C_{\phi}^{\sigma}$ induced on $S$. A \emph{matching} in a graph is a set of edges no two of which share a vertex.

\begin{definition}[Matching number and weighted exposure]
For a logical support $S$, the \emph{matching number}
\begin{equation}
\nu_{\phi}^{\sigma}(L)=\nu\Bigl(C_{\phi}^{\sigma}[S]\Bigr)
\end{equation}
is the size of a maximum matching in the support-induced subgraph, and the \emph{weighted exposure}
\begin{equation}
\Ew_{\phi}^{\sigma}(L)=\sum_{\{i,j\}\,\subseteq\, S} w_{ij}
\label{eq:exposure-def}
\end{equation}
is the total retained pair weight on the support.
\end{definition}

We now relate the matching number to the effective fault weight by treating each retained single- or pair-data location as a single elementary event.

\begin{definition}[Effective fault weight on a support]
Let $L$ be a sector logical operator with support $S$. The effective fault weight of $L$ on $S$, written $w^{\sigma}_{\mathrm{eff},\phi}(L)$, is the minimum number of retained elementary fault locations whose product acts as $L$ on $S$, counting one for each retained single location and one for each retained pair location.
\end{definition}

\subsection{Adversarial effective-distance reduction}

Let $d_{\sigma}$ denote the code distance in sector $\sigma$---the minimum weight of a nontrivial sector-$\sigma$ logical operator---and let $\mathcal L_{\sigma,d_\sigma}$ denote the set of sector-$\sigma$ logical operators of minimum weight $d_\sigma$. Define the sector effective distance
\begin{equation}
d^{\sigma}_{\mathrm{eff}}(\phi):=\min_L w^{\sigma}_{\mathrm{eff},\phi}(L),
\end{equation}
where the minimum is over all sector logical operators.

\begin{theorem}[Exact effective fault weight on a support]
\label{thm:deff}
Let $L$ be a sector logical operator with support $S$, and let $\nu_{\phi}^{\sigma}(L)$ be the matching number of $C_{\phi}^{\sigma}[S]$. Assume that every vertex of $S$ admits a retained single-location sector fault of the correct Pauli type. Then
\begin{equation}
w^{\sigma}_{\mathrm{eff},\phi}(L)= |S|-\nu_{\phi}^{\sigma}(L).
\label{eq:weff-bound}
\end{equation}
Consequently,
\begin{equation}
d^{\sigma}_{\mathrm{eff}}(\phi)\,\le\, d_{\sigma}-\max_{L\,\in\,\mathcal L_{\sigma,d_{\sigma}}}\nu_{\phi}^{\sigma}(L).
\label{eq:deff-bound}
\end{equation}
Equality holds in Eq.~\eqref{eq:deff-bound} if the minimizer of $w_{\mathrm{eff}}$ is attained among minimum-weight logicals, which is the case in all BB72 instances studied here.
\end{theorem}

\begin{proof}[Proof of Theorem~\ref{thm:deff}]
\emph{Upper bound.} Let $M$ be a maximum matching in $C_{\phi}^{\sigma}[S]$, so $|M|=\nu_{\phi}^{\sigma}(L)$. Use one retained pair fault for each edge in $M$ and one single-qubit fault for each unmatched vertex. The total is
\begin{equation}
|M|+\bigl(|S|-2|M|\bigr)=|S|-\nu.
\end{equation}

\emph{Lower bound.} Let $F$ be any set of retained pair edges used to realize $L$ on $S$. Every support vertex must be covered an odd number of times by the union of pair and single faults. Vertices of odd degree in $(S,F)$ require no additional single fault; all others require one. The minimum number of singles is therefore $|S|-o(F)$, where $o(F)$ is the number of odd-degree vertices in the subgraph $(S,F)$, and the total cost is $|F|+|S|-o(F)$. Decompose $(S,F)$ into nontrivial connected components $C_1,\,\ldots,\,C_t$ with $v_r$ vertices and $e_r$ edges each. In each component $o_r\le v_r$ and $e_r\ge v_r-1$ (connectivity), so $o_r-e_r\le 1$. Summing gives $o(F)-|F|\le t$. Selecting one edge from each component yields a matching of size $t$, hence $t\le\nu$. Therefore,
\begin{align}
|F|+|S|-o(F)&=|S|-\bigl(o(F)-|F|\bigr)\\
&\ge|S|-\nu.
\end{align}

Combining the two bounds gives $w^{\sigma}_{\mathrm{eff},\phi}(L)=|S|-\nu$. Restricting to minimum-weight logicals and taking the minimum yields Eq.~\eqref{eq:deff-bound}.

\end{proof}

\begin{remark}
The single-location assumption holds in the implemented BB circuit model, where each data qubit undergoes single-qubit depolarizing errors with rate $p$ per tick.
\end{remark}

Each matched support pair removes one elementary fault location from the adversarial count, so more crossings on the support mean a lower effective distance.

\begin{corollary}
\label{cor:complete-weff}
If $C_{\phi}^{\sigma}[S]$ is the complete graph on $S$, then $w^{\sigma}_{\mathrm{eff},\phi}(L)=\lceil|S|/2\rceil$.
\end{corollary}

Under the crossing kernel, Theorem~\ref{thm:deff} becomes a combinatorial statement: projecting the same-round crossings on the support directly lowers the effective-distance bound.

\subsection{Support saturation and weighted exposure for positive kernels}

\begin{proposition}[Support saturation for strictly positive kernels]
\label{prop:saturation-main}
Fix a sector logical operator $L$ with support $S$. Suppose that for every unordered pair $\{i,j\}\subseteq S$ there exists at least one retained round in which the corresponding routed gate blocks are simultaneously active, disjoint, and separated by a finite distance, and that the retained pair coefficient is strictly positive for every such finite same-round separation. Then the support-induced graph $C_{\phi}^{\sigma}[S]$ is complete. Consequently, once $|S|$ is fixed, the matching number depends only on support size and no longer distinguishes embeddings.
\end{proposition}

\begin{proof}[Proof of Proposition~\ref{prop:saturation-main}]
By assumption, for every unordered pair $\{i,j\}\subseteq S$, there is at least one retained same-round gate-block pair whose propagated coefficient contributes a strictly positive pair weight to the retained sector model. Hence, the aggregated retained coefficient $w_{ij}$ is strictly positive for every unordered support pair, so every unordered pair of vertices is an edge of $C_{\phi}^{\sigma}[S]$. Therefore $C_{\phi}^{\sigma}[S]$ is the complete graph on $S$, and its matching number depends only on $|S|$.
\end{proof}

The saturation hypothesis is a condition on the schedule rather than the embedding. We verify it directly on the BB72 pure-$q(L)$ weight-$6$ family of Sec.~\ref{sec:logical-aware}: for every support in the family, the three retained $B$ rounds together cover all $\binom{6}{2}=15$ support pairs, so under any strictly positive kernel the support graph is complete. The sector-relevant rounds on the BB144 weight-$12$ supports used in the scaling check cover the support pairs in the same way.

Proposition~\ref{prop:saturation-main} and Corollary~\ref{cor:complete-weff} imply that once the support graph saturates, every embedding has effective fault weight $\lceil|S|/2\rceil$ on that support. Finite-coupling discrimination then depends on weighted exposure rather than adversarial distance. To turn that statement into a risk bound, we need to specify how retained pair events co-occur. The standard stochastic noise model in BB memory simulations treats each elementary fault location as an independent Bernoulli trial~\cite{bravyi_memory_2024}, and we apply the same assumption to the retained pair locations.

\begin{assumption}[Independent retained pair events]
\label{asm:indep}
Within a single syndrome-extraction cycle, the retained pair indicators $Y_{ij}\sim\mathrm{Bernoulli}(w_{ij})$ on distinct support pairs $\{i,j\}\subseteq S$ are independent.
\end{assumption}

Assumption~\ref{asm:indep} is the only step beyond the controlled truncation of Theorem~\ref{thm:schedule-prop} on which the finite-coupling analysis below relies. Multi-pair coincidences in the same tick enter the full-tick twirl at $O(\Theta_t^4)$ and are counted as part of the discarded mass $M_{\ge 3}$; Assumption~\ref{asm:indep} is the declaration that we do not resum those higher-weight contributions back into the retained pair-event statistics.

For conciseness, we abbreviate the weighted exposure on the support of $L$ by $\Ew:=\Ew_\phi^\sigma(L)$ for the remainder of this subsection.

\begin{proposition}[Exposure as the first-order pair-risk parameter]
\label{prop:union}
Under Assumption~\ref{asm:indep}, the probability that at least one retained pair event occurs on the support of $L$ satisfies
\begin{equation}
\Pr\!\left[\sum_{\{i,j\}\,\subseteq\, S} Y_{ij}\ge 1\right]=1-\!\prod_{\{i,j\}\,\subseteq\, S}\left(1-w_{ij}\right),
\label{eq:exact-pair-risk}
\end{equation}
and is bounded by
\begin{equation}
1-e^{-\Ew}\,\le\,\Pr\!\left[\sum_{\{i,j\}\,\subseteq\, S} Y_{ij}\ge 1\right]\,\le\,\min\{1,\,\Ew\}.
\label{eq:two-sided}
\end{equation}
\end{proposition}

\begin{proof}
Let $Y_{ij}\sim\mathrm{Bernoulli}(w_{ij})$ be the retained pair indicators on support pairs $\{i,j\}\subseteq S$, independent under Assumption~\ref{asm:indep}. Equation~\eqref{eq:exact-pair-risk} is then the complement of the all-clear probability $\prod(1-w_{ij})$. The upper bound in Eq.~\eqref{eq:two-sided} is the union bound. The lower bound follows from $1-w_{ij}\le e^{-w_{ij}}$ and multiplication. In the weak-pair regime $\Ew\ll 1$, expanding $1-\prod(1-w_{ij})=\sum w_{ij}+O\bigl((\sum w_{ij})^2\bigr)$ gives a first-order estimate $\Ew+O(\Ew^2)$.
\end{proof}

Under the retained independent-pair model, weighted exposure is therefore the first finite-coupling quantity that separates embeddings once support graphs saturate: lower total pair weight on a logical support means lower pair-event risk on that support.

\begin{corollary}[Certified improvement under the logical-aware objective]
\label{cor:certified-Jkappa}
For a reference family $\RX$ of logical supports, define
\begin{equation}
J_{\kappa}(\phi;\,\RX):=\max_{L\,\in\,\RX}\,\Ew_{\phi}^{\sigma}(L)
\end{equation}
and the worst-case pair-event probability over $\RX$,
\begin{equation}
B_{\kappa}(\phi;\,\RX) :=
\max_{L\,\in\,\RX}\,\Pr\!\left[\sum_{\{i,j\}\,\subseteq\,\supp(L)} Y_{ij}\ge 1\right].
\end{equation}
Then $B_{\kappa}(\phi;\,\RX)\le J_{\kappa}(\phi;\,\RX)$. In particular, if $J_{\kappa}(\phi';\,\RX)<J_{\kappa}(\phi;\,\RX)$, then the upper bound on the worst-case pair-event burden over $\RX$ is strictly smaller for $\phi'$ than for $\phi$.
\end{corollary}

\begin{proof}
Apply Proposition~\ref{prop:union} supportwise and take the maximum over $L\in\RX$.
\end{proof}

\subsection{Crossing-kernel evaluation on the BB72 reference support}

For the BB72 reference support $S_{\mathrm{ref}}$ of Sec.~\ref{sec:ref-support}, the crossing criterion Eq.~\eqref{eq:cross-criterion} yields \bbSeventyTwoCrossSupportCrossings{} crossing edges on $S_{\mathrm{ref}}$ across the three relevant $B$ rounds in the monomial embedding, with maximum matching number $\nu_\phi^X(L_{\mathrm{ref}})=\bbSeventyTwoCrossMatching$. Theorem~\ref{thm:deff} then gives
\begin{equation}
w_{\mathrm{eff}}(L_{\mathrm{ref}})\le 6-3=3.
\end{equation}
In the biplanar implementation, the crossing kernel induces no edges on $S_{\mathrm{ref}}$, so Theorem~\ref{thm:deff} leaves the bound at $w_{\mathrm{eff}}(L_{\mathrm{ref}})\le 6$.

\subsection{Weighted exposure under positive kernels}

For the worked BB72 support, the algebraic audit gives
\begin{equation}
\Ew^X_{\phi_{\mathrm{mono}}}(L_{\mathrm{ref}})=\bbSeventyTwoMonoExposure{},
\qquad
\Ew^X_{\phi_{\mathrm{bi}}}(L_{\mathrm{ref}})=\bbSeventyTwoBiExposure{}
\end{equation}
for the regularized algebraic kernel with $\alpha=3$, $r_0=1$, and $J_0\tau=0.04$. The matching numbers agree, but the monomial exposure is larger, so the weighted exposure distinguishes the embeddings where the matching numbers do not.

The biplanar embedding also admits a simple analytic upper bound. Let $\delta_A$ and $\delta_B$ denote the minimum separations between support pairs in the same round, in the layers carrying $B_3$ and $B_{1,2}$ respectively. For the worked BB72 support,
\begin{equation}
\Ew^X_{\phi_{\mathrm{bi}}}(L_{\mathrm{ref}})
\le 15\left(2q_{\kappa}(\delta_B)+q_{\kappa}(\delta_A)\right),
\label{eq:bi-surrogate-proof}
\end{equation}
with $q_\kappa(\cdot)$ the retained sector coefficient of Eq.~\eqref{eq:q-kappa-def}. The prefactor charges each of the $\binom{6}{2}=15$ support pairs with the worst-case layer separation in every round, whereas a given pair contributes in only one of the three relevant $B$ rounds (two in layer $G_B$ and one in $G_A$); the bound therefore trades tightness for a closed analytic form.

\subsection{Two-dimensional summability and the AKP-type compatibility}

The summability and AKP statements in this subsection concern the amplitude-level microscopic Hamiltonian of Eq.~\eqref{eq:tick-ham-main} before twirling, and give asymptotic compatibility conditions for regularized algebraic kernels.

\begin{lemma}[Two-dimensional summability for regularized power laws]
\label{lem:summability}
Assume a planar layout with minimum pairwise separation $a>0$. Then for any data location $i$ and any $\alpha>2$,
\begin{equation}
\sum_{j\neq i}\left(1+\frac{d_{ij}}{r_0}\right)^{-\alpha}
\le 24\,\zeta\!\left(\alpha-1\right)\left(\frac{r_0}{a}\right)^{\alpha},
\label{eq:summability}
\end{equation}
where $\zeta$ is the Riemann zeta function.
\end{lemma}

\begin{proof}
Center annuli of width $a$ around $i$ and let $\mathcal{A}_n$ be the annulus with inner radius $na$ and outer radius $(n+1)a$. Placing non-overlapping disks of radius $a/2$ at each data location, and noting that a disk centered in $\mathcal{A}_n$ lies in the fattened annulus from $(n-\tfrac12)a$ to $(n+\tfrac32)a$ (area $\pi(4n+2)a^2$), gives at most $16n+8\le 24n$ points in $\mathcal{A}_n$ for $n\ge 1$. Hence
\begin{equation}
\sum_{j\neq i}\left(1+\frac{d_{ij}}{r_0}\right)^{-\alpha}
\le 24\left(\frac{r_0}{a}\right)^{\alpha}\sum_{n=1}^{\infty}n^{1-\alpha},
\end{equation}
and the final series equals $\zeta(\alpha-1)$, which converges if and only if $\alpha>2$. A tighter packing argument replaces the prefactor $24$ with $8$; convergence itself uses only the $O(n)$ shell count.
\end{proof}

The exponent $\alpha=2$ is the planar threshold for summability of algebraic couplings. Lemma~\ref{lem:summability} is stated for data locations rather than routed gate blocks; it carries over to the block-level quantity $\eta_e^{(t)}(\phi)$ of Eq.~\eqref{eq:eta-main} under the standard bounded-density routing assumption, namely that each routed block is associated with a bounded number of data locations and has a minimum closest-approach distance $a$ from every other simultaneously active block. Under this assumption, define the worst-case tickwise block exposure
\begin{equation}
\eta_{\max}(\phi):=\sup_t\max_{e\,\in\, E_t}\eta_e^{(t)}(\phi),
\label{eq:eta-max}
\end{equation}
with $\eta_e^{(t)}(\phi)$ the same-tick block exposure of Eq.~\eqref{eq:eta-main}. For a regularized algebraic kernel, Lemma~\ref{lem:summability} gives
\begin{equation}
\eta_{\max}(\phi)\le 24\,\zeta\!\left(\alpha-1\right)\left(\frac{r_0}{a}\right)^{\alpha},
\end{equation}
up to a bounded-density constant.

\begin{theorem}[AKP pairwise long-range criterion, restated from Ref.~\onlinecite{akp_2006}]
There exists $\eta_0>0$ such that fault-tolerant simulation is possible whenever the microscopic Hamiltonian can be written as a sum of pair terms obeying
\begin{equation}
\sup_t\max_i\sum_{j\,\neq\, i}\lVert H_{ij}(t)\rVert\,\tau\,<\,\eta_0,
\end{equation}
with $\lVert\cdot\rVert$ the operator norm.
\end{theorem}

Each pair term in Eq.~\eqref{eq:tick-ham-main} has operator norm $J_0\kappa(d_\phi(e,e'))$, so the aggregate coupling per gate block is bounded by $J_0\eta_e^{(t)}(\phi)$. A sufficient AKP-type condition at the gate-block level is therefore
\begin{equation}
J_0\tau\,\eta_{\max}(\phi)\,<\,\eta_0^{\mathrm{blk}},
\end{equation}
with $\eta_0^{\mathrm{blk}}$ the AKP threshold constant renormalized to constant-size gate blocks. The condition applies to $\Hxt$ before twirling. Lemma~\ref{lem:summability} assumes an idealized layout with fixed positive minimum separation $a>0$; the finite biplanar layouts of the BB72 and BB144 numerics are evaluated directly by the geometry engine.

\section{Logical-aware single-layer design program}
\label{sec:logical-aware}

The positive-kernel results of the previous section identify weighted exposure as the finite-coupling quantity that distinguishes embeddings. To turn that into a design problem, we restrict to logical operators whose support lies entirely on the $q(L)$ register---the \emph{pure-$q(L)$} operators---which on BB72 can be enumerated exhaustively.

We allow permutations $\sigma_L,\sigma_Z\in S_M$ of the $q(L)$ and $q(Z)$ row orders, where $S_M$ is the symmetric group on $M$ indices introduced in Sec.~\ref{sec:model}. A transposition $\tau_{ij}\in S_M$ swaps indices $i$ and $j$. For BB codes written in the standard CSS matrix form
\begin{equation}
H_X=\left[A\;B\right],\qquad H_Z=\left[B^\top\;A^\top\right],
\end{equation}
consider $X$-type operators whose $q(L)\oplus q(R)$ support is represented by $[u\;v]^\top\in\F^{2M}$, where $u\in\F^M$ acts on $q(L)$ and $v\in\F^M$ acts on $q(R)$. An operator with $v=0$ is called \emph{pure-$q(L)$}.

\begin{proposition}[Exact pure-$q(L)$ quotient]
\label{prop:pureL-main}
A pure-$q(L)$ operator is a valid $X$-sector logical operator if and only if $u\in\kerop B^\top$, and its trivial pure-$q(L)$ representatives form the subgroup
\begin{equation}
\TL=\{A^\top\lambda\mid B^\top\lambda=0\}.
\end{equation}
Hence pure-$q(L)$ $X$ logical classes are parameterized by the quotient
\begin{equation}
\kerop(B^\top)/\TL.
\label{eq:pureL-quotient-main}
\end{equation}
\end{proposition}

\begin{proof}[Proof of Proposition~\ref{prop:pureL-main}]
The operator commutes with every $Z$ stabilizer if and only if
\begin{equation}
H_Z[u\;0]^\top=B^\top u=0,
\end{equation}
so $u\in\kerop B^\top$. It is trivial as a logical operator if and only if $[u\;0]^\top$ lies in the column space of $H_X^\top$, that is, if and only if there exists $\lambda\in\F^M$ such that
\begin{equation}
\begin{bmatrix}u\\0\end{bmatrix}=H_X^\top\lambda=\begin{bmatrix}A^\top\lambda\\B^\top\lambda\end{bmatrix}.
\end{equation}
The pure-$q(L)$ requirement forces $B^\top\lambda=0$, so the trivial pure-$q(L)$ representatives are $u=A^\top\lambda$ with $B^\top\lambda=0$. The quotient follows.
\end{proof}

A naive quotient by the column space of $A^\top$ would lose the constraint $B^\top\lambda=0$, so the full quotient $\kerop(B^\top)/\TL$ is required.

\subsection{BB72 family}

For BB72, the exhaustive audit used in this work gives
\begin{equation}
\dim \kerop B^\top=12,
\qquad
\dim \TL=6,
\end{equation}
so the pure-$q(L)$ class space has dimension $6$. Exhaustive enumeration yields \bbSeventyTwoPureLFamilyCount{} weight-$6$ pure-$q(L)$ supports, which form the reference family $\RX$ for the BB72 logical-aware objective. For BB108, the literature code parameters are \bb{108,8,10}, but the minimum weight among pure-$q(L)$ representatives is $12$; we keep those two quantities separate throughout.

\subsection{Logical-aware objectives and finite termination}

Given a reference family $\RX$ of minimum-weight pure-$q(L)$ $X$ logical supports, define the logical-aware objective for single-layer embeddings by
\begin{equation}
J_{\kappa}(\phi;\,\RX)=\max_{L\,\in\,\RX}\Ew_{\phi}^{X}(L).
\label{eq:Jkappa}
\end{equation}
For the crossing kernel, one may likewise define
\begin{equation}
J_{\times}(\phi;\,\RX)=\max_{L\,\in\,\RX}\nu_{\phi}^{X}(L),
\end{equation}
but in the positive-kernel regime relevant here, $J_{\kappa}$ carries more information. By Corollary~\ref{cor:certified-Jkappa}, lowering $J_{\kappa}$ lowers a strict upper bound on pair-event incidence on each support over the entire reference family.

Let $\Phi_4$ denote the finite set of admissible four-column row-permutation embeddings. The logical-aware design target is then
\begin{equation}
\phi_{\mathrm{LA}}^{\star}\in \argmin_{\phi\,\in\,\Phi_4} J_{\kappa}(\phi;\,\RX).
\label{eq:exact-la-opt}
\end{equation}
The production search does not solve Eq.~\eqref{eq:exact-la-opt} globally. It uses multirestart simulated annealing only to generate warm starts, and then applies deterministic best-improving two-swap descent on $J_\kappa$ until no improving transposition remains.

\begin{proposition}[Finite termination of deterministic swap descent]
\label{thm:swap-descent}
The deterministic logical-aware two-swap descent terminates after finitely many accepted moves and outputs a two-swap local minimum of $J_{\kappa}$ on $\Phi_4$.
\end{proposition}

\begin{proof}
The state space $\Phi_4$ is finite because it is a finite product of permutation groups. Every accepted move strictly decreases the real-valued objective $J_{\kappa}$, so no state can be revisited. Hence, the descent terminates after finitely many accepted moves. The stopping rule is that no transposition of either $\sigma_L$ or $\sigma_Z$ lowers $J_{\kappa}$ further; therefore, the returned embedding is a two-swap local minimum.
\end{proof}

\begin{algorithm}[!htbp]
\caption{Logical-aware embedding search}\label{alg:la-search}
\KwIn{BB matrices $(A,B)$; reference family $\RX\subseteq 2^{[M]}$; kernel $\kappa$; restarts $n_r$}
\KwOut{Two-swap local minimizer $(\sigma_L^\star,\sigma_Z^\star)\in S_M\times S_M$; value $J_\kappa(\sigma_L^\star,\sigma_Z^\star)$}
\BlankLine
$J_{\kappa}(\sigma_L,\sigma_Z)\gets\max_{L\in\RX}\Ew^{X}_{\phi(\sigma_L,\sigma_Z)}(L)$\;
$\mathcal V_2(\sigma_L,\sigma_Z)\gets\{(\sigma_L\!\circ\!\tau_{ij},\sigma_Z)\}_{i<j}\cup\{(\sigma_L,\sigma_Z\!\circ\!\tau_{ij})\}_{i<j}$\;
\BlankLine
\For{$r=1,\ldots,n_r$}{
$(\sigma_L^{(r)},\sigma_Z^{(r)})\gets$ uniform random element of $S_M\times S_M$;
Run simulated annealing on $\mathcal V_2$ from $(\sigma_L^{(r)},\sigma_Z^{(r)})$;
}
$(\sigma_L^\star,\sigma_Z^\star)\gets\argmin_{1\le r\le n_r} J_\kappa(\sigma_L^{(r)},\sigma_Z^{(r)})$\;
\BlankLine
\While{\textnormal{\textbf{true}}}{
$s^\star\gets\argmin_{s\in\mathcal V_2(\sigma_L^\star,\sigma_Z^\star)}J_\kappa(s)$\;
\lIf{$J_\kappa(s^\star)\ge J_\kappa(\sigma_L^\star,\sigma_Z^\star)$}{\textbf{break}}
$(\sigma_L^\star,\sigma_Z^\star)\gets s^\star$\;
}
\end{algorithm}

Algorithm~\ref{alg:la-search} is the search procedure used in the numerical study. Simulated annealing supplies the warm start, and Proposition~\ref{thm:swap-descent} applies to the deterministic descent stage. In the implementation, every objective evaluation uses routed-segment distances from the geometry engine without approximation, and every crossing-local diagnostic computes the matching number with a general-graph maximum-matching routine rather than a bipartite shortcut. On the BB72 weight-$6$ pure-$q(L)$ family, the deterministic audit yields
\begin{equation}
J_\kappa(\phi;\,\RX)=\max_{L\in\RX}\Ew_\phi^X(L)=\begin{cases}
\bbSeventyTwoLAMaxExposureMono{} & \phi=\phi_{\mathrm{mono}},\\
\bbSeventyTwoLAMaxExposureLA{} & \phi=\phi_{\mathrm{LA}},\\
\bbSeventyTwoLAMaxExposureBi{} & \phi=\phi_{\mathrm{bi}}.
\end{cases}
\end{equation}
The logical-aware embedding therefore reduces the monomial value by \bbSeventyTwoLAImprovementPct\%, while the biplanar reference reduces it by \bbSeventyTwoBiImprovementPct\%.

By Corollary~\ref{cor:certified-Jkappa}, a lower $J_{\kappa}$ gives a tighter upper bound on pair-event incidence on each support over the chosen logical family. The numerical tests below track the same ordering in logical error rate.

Appendix~\ref{app:decoder} develops a refinement $C_D(\phi)$ that sharpens $J_\kappa$ at first order by incorporating the decoder's sensitivity to individual pair locations. A pilot estimate at one operating point confirms the same embedding ordering with a sharper monomial-to-biplanar ratio ($3.2\times$ for $C_D$ versus $1.8\times$ for $J_\kappa$). The optimization target throughout the rest of this paper remains $J_\kappa$, because $J_\kappa$ is intrinsic to the retained model and does not depend on a decoder.

\section{Computational validation}
\label{sec:results}

This section uses two geometry metrics. The reference-support exposure $\Ew_\phi^X(L_{\mathrm{ref}})$ on the fixed BB72 support $S_{\mathrm{ref}}$ (Sec.~\ref{sec:ref-support}) drives the matching-number test and the BB72 distance-decay and scatter tests. The family-wise objective $J_\kappa(\phi;\RX)=\max_{L\in\RX}\Ew_\phi^X(L)$ (Sec.~\ref{sec:logical-aware}) drives the many-layout and logical-aware tests. The first compares embeddings on a single support, while the second evaluates worst-case behavior over the optimized logical family.

The computational program has a single primary validation target, BB72, and a smaller secondary check, BB144. BB90 and BB108 appear only as supporting slices in Appendix~\ref{app:numerics}.

\subsection{Simulation protocol}
\label{sec:syndrome}

The circuit-level noise model applies a two-qubit depolarizing channel (rate $p_{\mathrm{cnot}}$) after every active \textsc{cnot}, a one-qubit depolarizing channel (rate $p_{\mathrm{idle}}$) on every idle qubit, single-Pauli preparation errors (rate $p_{\mathrm{prep}}$), and single-Pauli measurement errors (rate $p_{\mathrm{meas}}$). Geometry-induced pair channels are controlled separately by $(J_0,\tau,\kappa)$ and act on the sector-relevant data pairs of Sec.~\ref{sec:model}.

Each production sample runs the full depth-$8$ BB cycle of Ref.~\onlinecite{bravyi_memory_2024} with this noise placement. Single and pair data faults are injected at the rates given by the linearized retained coefficients of Corollary~\ref{cor:linearized}: for each simultaneously active pair of sector-relevant gate blocks at separation $d$, a correlated Pauli fault is applied on both data qubits with probability $q_\kappa^{\mathrm{lead}}(d):=\sin^2\tau J_0\kappa(d)$, which equals the single-pair retained coefficient up to $O(\Theta_t^4)$ corrections from multi-pair coincidences (Theorem~\ref{thm:schedule-prop}). Circuit sampling uses Stim~\cite{gidney_stim_2021}; repeated raw runs at identical operating points are merged before plotting. In all sweeps the local rates are tied,
\begin{equation}
p_{\mathrm{cnot}}=p_{\mathrm{idle}}=p_{\mathrm{prep}}=p_{\mathrm{meas}}=:p,
\end{equation}
so $p$ denotes a common circuit-level error rate. Embeddings are compared at fixed $p$. The main BB72 and BB144 line sweeps use $5{,}000$ shots per raw run, and merged operating points accumulate higher counts (up to $15{,}000$ where pilot and window runs overlap). All reported confidence intervals (CIs) are 95\% binomial. Appendix~\ref{app:numerics} includes an untied-rate robustness check confirming that the hierarchy is not specific to this insertion rule.

The baseline decoder is BP+OSD as in Ref.~\onlinecite{roffe_bp_osd_2020}. The source archive also contains a correlation-aware BP+OSD prototype, but it targets the retained single-and-pair model rather than the full circuit-level channel of the multi-round simulation, so it is not a matched comparison at circuit level. Appendix~\ref{app:decoder} records the retained-model decoder constructions.

After merging repeated runs at identical control parameters, the dataset contains $\nSemanticTotal{}$ operating points: $\nBBSeventyTwoPoints{}$ on BB72, $\nBBOneFortyFourPoints{}$ on BB144, and $\nBBNinetyPoints{}$ and $\nBBOneOhEightPoints{}$ on BB90 and BB108. The main text reports the decoded $X$ sector analyzed in Sec.~\ref{sec:theory}. The number of syndrome cycles is fixed to the benchmark distance of the simulated code, namely $6$ for BB72 and $12$ for BB144. When an operating point has zero observed logical failures, the line-sweep plots place an open marker at the 95\% upper bound rather than at zero on the logarithmic axis.

\subsection{Matching-number claim on BB72}
\label{subsec:crossing}

Under the crossing kernel, Theorem~\ref{thm:deff} predicts that the matching-number mechanism on the logical support lowers the sector effective distance bound. The analytical evaluation on $S_{\mathrm{ref}}$ in Sec.~\ref{sec:theory} gives $w_{\mathrm{eff}}\le \bbSeventyTwoCrossDeffMono{}$ for the monomial embedding and $w_{\mathrm{eff}}\le \bbSeventyTwoCrossDeffBi{}$ for the biplanar embedding (cf.\ Fig.~\ref{fig:bridge}).

Figure~\ref{fig:bb72-crossing} shows the corresponding crossing-kernel logical-error-rate sweep on BB72. At $J_0\tau=0.04$ and $p=10^{-3}$, the monomial embedding has logical error rate \bbSeventyTwoCrossMonoLER{} with 95\% CI [\bbSeventyTwoCrossMonoLERLo{}, \bbSeventyTwoCrossMonoLERHi{}], while the biplanar embedding has zero failures in the present sample, corresponding to a 95\% upper confidence bound of \bbSeventyTwoCrossBiLERHi{}. Under this diagnostic kernel, the matching mechanism of Theorem~\ref{thm:deff} produces a logical penalty consistent with the halved effective-distance bound.

\begin{figure}[!htbp]
\centering
\includegraphics[width=\columnwidth]{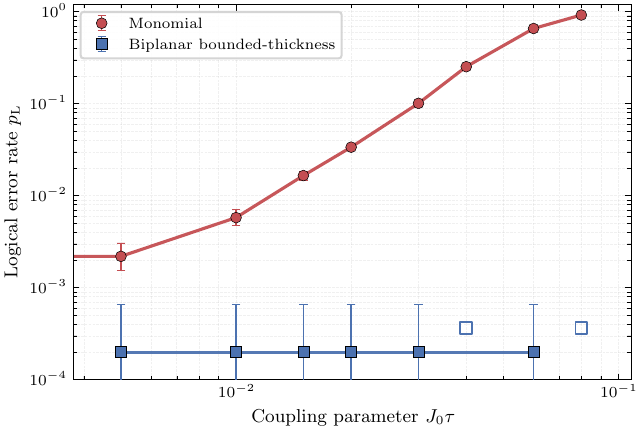}
\caption{BB72 crossing-kernel diagnostic sweep in the $X$ sector at $p=10^{-3}$ and $6$ cycles. Error bars are 95\% CIs; open markers denote zero-failure operating points, plotted at the 95\% upper bound.}
\label{fig:bb72-crossing}
\end{figure}

\subsection{Exposure-ordering claim on BB72}
\label{subsec:exposure-ordering}

Under any strictly positive kernel, Propositions~\ref{prop:saturation-main} and~\ref{prop:union} identify weighted exposure as the leading discriminator of \emph{pair-event burden} between embeddings. The three tests below ask whether the same quantity also tracks the logical error rate, varying in turn the coupling and decay exponent on the two reference embeddings, the kernel family across all sampled operating points, and the embedding itself across a random ensemble of single-layer layouts.

\paragraph*{Distance-decay sweeps.}
The power-law kernel is the main distance-decay profile here, because it preserves strictly positive geometry dependence and remains analytically summable for $\alpha>2$. Figure~\ref{fig:bb72-power} shows the two BB72 sweeps: coupling strength at fixed $\alpha=3$, and decay exponent at fixed $J_0\tau=0.04$. The monomial embedding has the larger logical error rate across the sweep. At the reference operating point $(J_0\tau,\alpha,p)=(0.04,3,10^{-3})$, the monomial embedding has logical error rate \bbSeventyTwoPowerMonoLER{} with 95\% CI [\bbSeventyTwoPowerMonoLERLo{}, \bbSeventyTwoPowerMonoLERHi{}], whereas the biplanar embedding has \bbSeventyTwoPowerBiLER{} with 95\% CI [\bbSeventyTwoPowerBiLERLo{}, \bbSeventyTwoPowerBiLERHi{}], a factor of \bbSeventyTwoPowerLERRatio{} at fixed code, schedule, decoder, and local noise. The analytical geometry metrics produce the same ordering. On $S_{\mathrm{ref}}$, the weighted exposure under the regularized algebraic kernel is \bbSeventyTwoMonoExposure{} for the monomial embedding and \bbSeventyTwoBiExposure{} for the biplanar embedding, and the maximum aggregated retained pair probability seen by a data location is also larger in the monomial embedding (\bbSeventyTwoPowerAggProbMono{}) than in the biplanar embedding (\bbSeventyTwoPowerAggProbBi{}).

\begin{figure*}[!htbp]
\centering
\includegraphics[width=\textwidth]{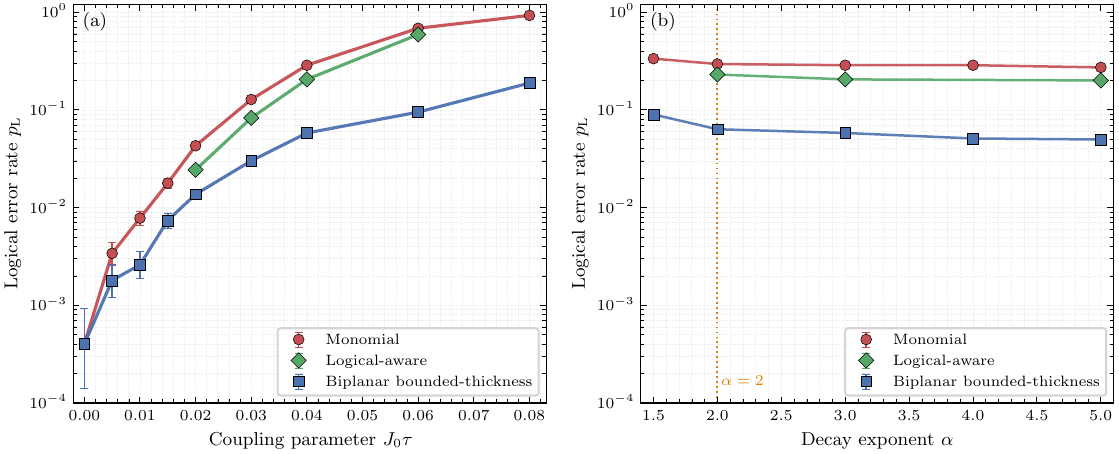}
\caption{BB72 $X$-sector logical error rate under the regularized power-law kernel. (a) Sweep in $J_0\tau$ at fixed $\alpha=3$ and $p=10^{-3}$. (b) Sweep in $\alpha$ at fixed $J_0\tau=0.04$ and $p=10^{-3}$. The monomial and biplanar embeddings are shown on the full sampled grids; the logical-aware embedding is shown on the available subsets $J_0\tau\in\{0.02,0.03,0.04,0.06\}$ (a) and $\alpha\in\{2,3,5\}$ (b).}
\label{fig:bb72-power}
\end{figure*}

\paragraph*{Exposure--LER scatter across kernels.}
A \emph{baseline operating point} is a tuple $(p,\kappa,J_0\tau,\phi)$ at which the monomial and biplanar embeddings are both sampled; points involving the logical-aware embedding are excluded from the baseline set. Figure~\ref{fig:scatter} plots the BB72 reference-support exposure $\Ew_\phi^X(L_{\mathrm{ref}})$ against the observed logical error rate for \baselineCorrelationN{} such points spanning the crossing, power-law, and exponential kernels. The baseline data give Spearman rank correlation $\rho_\mathrm{S}=\baselineSpearmanRho$ ($p$-value $=\baselineSpearmanP$). The scatter shows several roughly parallel branches rather than a single curve, because $p_L$ also depends on the local depolarizing floor $p$ and the kernel family $\kappa$. Within each kernel family the correlation is tighter, namely $\rho_\mathrm{S}=1.000$ (crossing, $n=15$), $0.923$ (exponential, $n=12$), and $0.893$ (power-law, $n=90$), all with $p$-value $<10^{-4}$. Restricting to fixed $p=10^{-3}$ gives $\rho_\mathrm{S}=0.965$ ($n=59$). At matched operating points where two embeddings share the same $(p,\kappa,J_0\tau)$ and differ only in layout, 64 of 65 non-tied pairs are concordant. Appendix~\ref{app:numerics} (Fig.~\ref{fig:bb72-by-kernel}) splits the scatter by kernel class. The logical-aware points, shown as stars and excluded from the quoted $\rho$, lie on the same trend.

\begin{figure}[!htbp]
\centering
\includegraphics[width=\columnwidth]{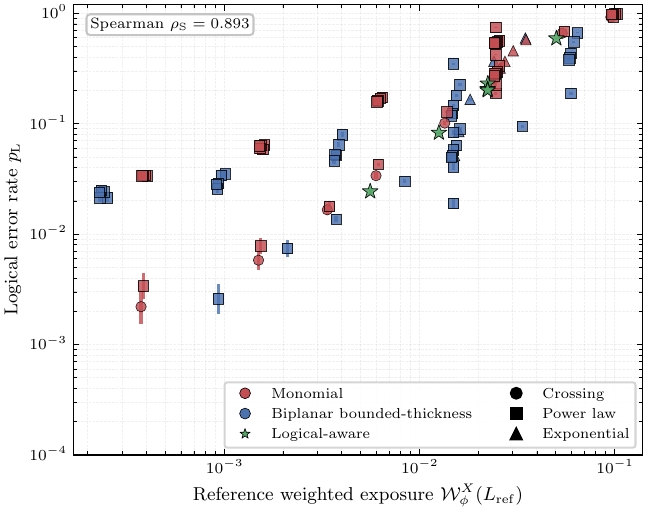}
\caption{BB72 reference-support exposure versus observed logical error rate. Colors indicate embedding family, and marker shape indicates kernel family. The quoted Spearman coefficient uses the \baselineCorrelationN{} baseline points with at least ten failures and excludes the logical-aware points.}
\label{fig:scatter}
\end{figure}

\paragraph*{Exposure ordering across many layouts.}
To extend the test beyond the two reference embeddings, we simulated $22$ distinct single-layer BB72 layouts at the fixed operating point $(p,\alpha,J_0\tau)=(10^{-3},3,0.04)$: the monomial layout, the logical-aware layout, and $20$ random row-permutation layouts. Each layout used $10{,}000$ shots under identical code, schedule, decoder, and local noise. Figure~\ref{fig:many-layout} plots the maximum family exposure $J_\kappa$ against the observed logical error rate for all $22$ layouts. The logical-aware layout ($J_\kappa=0.0224$, $p_L=0.212$) and the monomial layout ($J_\kappa=0.0303$, $p_L=0.284$) have the two lowest exposures and the two lowest logical error rates. All $20$ random layouts have higher exposures ($J_\kappa\in[0.043,0.054]$) and higher logical error rates ($p_L\in[0.44,0.62]$). The Spearman correlation across all $22$ layouts is $\rho_\mathrm{S}=0.552$ (Spearman $p$-value $=7.8\times10^{-3}$); the moderate value reflects scatter within the random ensemble, which occupies a narrow exposure band. The separation between optimized and random layouts is complete, so weighted exposure separates good layouts from bad ones even when, within a narrow exposure band, it does not fully order them.

\begin{figure}[!htbp]
\centering
\includegraphics[width=\columnwidth]{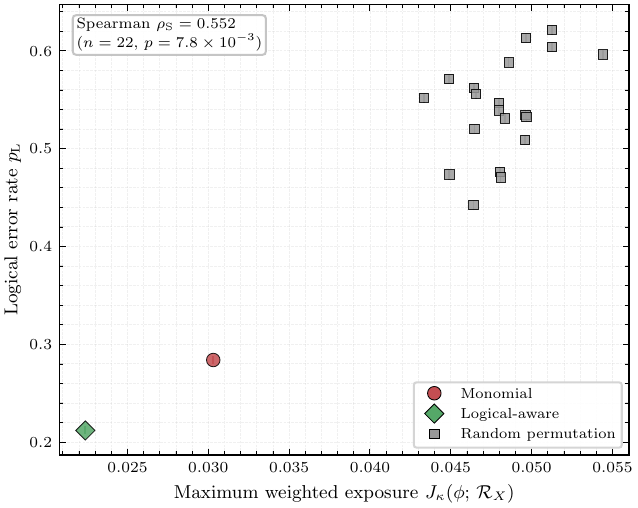}
\caption{Many-layout BB72 validation at fixed $(p,\alpha,J_0\tau)=(10^{-3},3,0.04)$. Each point is one of $22$ single-layer four-column layouts ($10{,}000$ shots, BP+OSD). Colored markers: monomial and logical-aware layouts. Gray squares: $20$ random row-permutation layouts.}
\label{fig:many-layout}
\end{figure}

\subsection{Logical-aware claim on BB72}
\label{subsec:la}

\begin{figure*}[!tbp]
\centering
\includegraphics[width=\textwidth]{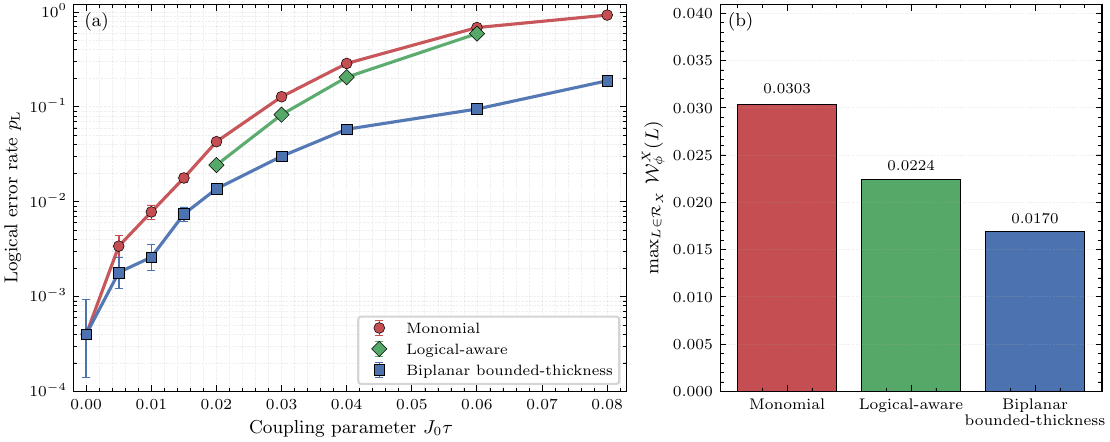}
\caption{BB72 logical-aware design validation in the $X$ sector for the power-law kernel at $\alpha=3$ and $p=10^{-3}$. (a) Stochastic coupling window $J_0\tau\in\{0.02,0.03,0.04,0.06\}$. (b) Deterministic max-exposure audit over the weight-$6$ pure-$q(L)$ family.}
\label{fig:la-window}
\end{figure*}

Corollary~\ref{cor:certified-Jkappa} predicts that minimizing the logical-aware objective $J_\kappa$ over the pure-$q(L)$ quotient bounds the worst-support pair-event incidence. The numerical question is whether the same minimization also lowers the observed logical error rate. On the BB72 weight-$6$ pure-$q(L)$ family, the monomial embedding has maximum exposure \bbSeventyTwoLAMaxExposureMono{}, the logical-aware embedding has \bbSeventyTwoLAMaxExposureLA{}, and the biplanar embedding has \bbSeventyTwoLAMaxExposureBi{}. The logical-aware embedding therefore reduces the worst-case exposure by \bbSeventyTwoLAImprovementPct\% relative to the monomial baseline. Across the full 36-support family, the mean exposure drops from $0.0212$ (monomial) to $0.0187$ (logical-aware) to $0.0120$ (biplanar), and the min-to-max range narrows from $[0.0096, 0.0303]$ to $[0.0113,0.0224]$ under the logical-aware layout. The raw crossing count can increase ($542$ vs.\ $522$), confirming that exposure, not crossings, is the optimization target.

The BB72 stochastic window in Fig.~\ref{fig:la-window} follows the equivalent analytical improvement. At $J_0\tau=0.04$, the logical-aware layout yields \bbSeventyTwoPowerLAFourLER{} with 95\% CI [\bbSeventyTwoPowerLAFourLERLo{}, \bbSeventyTwoPowerLAFourLERHi{}], compared with \bbSeventyTwoPowerMonoLER{} for the monomial layout and \bbSeventyTwoPowerBiLER{} for the biplanar layout. Across the four tested couplings $J_0\tau\in\{0.02,0.03,0.04,0.06\}$, the logical-aware embedding lies strictly between the monomial and biplanar curves, with total logical-error-rate reductions relative to monomial of approximately $43\%$, $35\%$, $28\%$, and $14\%$, respectively. The gain is largest in the moderate-correlation window and weakens as the monomial curve approaches saturation.

Panel~(b) of Fig.~\ref{fig:bb72-power} shows the same pattern across the tested exponent range $\alpha\in\{2,3,5\}$ at fixed $(J_0\tau,p)=(0.04,10^{-3})$: the logical-aware embedding remains strictly between the monomial and biplanar embeddings at every tested $\alpha$. The finite-coupling gain is smaller than the biplanar gain, in line with the deterministic exposure ordering.

\subsection{Scaling check on BB144}
\label{subsec:bb144}

Figure~\ref{fig:bb144} shows the analogous BB144 $J_0\tau$ and physical-error-rate sweeps in the $X$ sector. The same embedding hierarchy holds.

The geometry penalty matches the BB72 ordering and is larger at the reference operating point. At $(J_0\tau,\alpha,p)=(0.04,3,10^{-3})$, the monomial embedding has logical error rate \bbOneFortyFourPowerMonoLER{} with 95\% CI [\bbOneFortyFourPowerMonoLERLo{}, \bbOneFortyFourPowerMonoLERHi{}], whereas the biplanar embedding has \bbOneFortyFourPowerBiLER{} with 95\% CI [\bbOneFortyFourPowerBiLERLo{}, \bbOneFortyFourPowerBiLERHi{}], a factor of \bbOneFortyFourPowerLERRatio{}. The BB144 dataset is smaller than the BB72 dataset and is reported only as a check on the embedding hierarchy at the next benchmark size.

\begin{figure*}[!htbp]
\centering
\includegraphics[width=\textwidth]{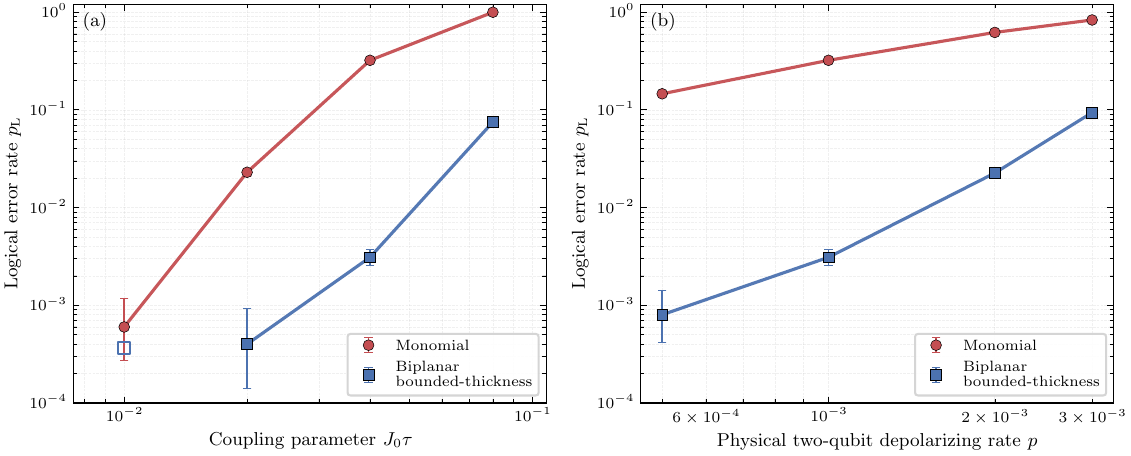}
\caption{BB144 $X$-sector scaling check. (a) Sweep in $J_0\tau$ at fixed $\alpha=3$ and $p=10^{-3}$. (b) Sweep in $p$ at fixed $J_0\tau=0.04$ and $\alpha=3$. Open markers denote 95\% upper bounds for zero-failure points.}
\label{fig:bb144}
\end{figure*}

Appendix~\ref{app:numerics} collects additional BB72 kernel sweeps, the BB72 phase-diagram and exposure-by-kernel diagnostics, the BB90/BB108 supporting slices, and a direct biplanar scaling comparison.

\section{Discussion and outlook}
\label{sec:discussion}

With the code and extraction schedule fixed, routed geometry changes the leading correlated-fault structure, and the logical performance with it. The model developed here keeps only the inter-block interaction component of a same-tick perturbation. Within that model, the monomial embedding incurs a higher worst-case pair-event burden and a higher logical error rate than the biplanar embedding across the explored BB72 and BB144 parameter window. A logical-aware two-swap local search over single-layer embeddings reduces both quantities relative to the monomial baseline.

The results depend on two modeling choices. The first is the phenomenological kernel $\kappa(d)$, which is constrained but not identified by public superconducting-hardware data (Appendix~\ref{app:hardware}). The second is the independent-pair treatment of the retained coefficients. Given those two choices, the remaining derivation---twirl, retained channel, support graph, and exposure metric---is closed in form. The retained model also admits a maximum-a-posteriori (MAP) decoder reduction (Appendix~\ref{app:decoder}), so geometry-aware priors can be propagated to the decoder.

These results align with the circuit-centric view of fault tolerance developed in recent work on residual-error metrics and spacetime codes~\cite{strikis_sec_2026,aitchison_beri_2025,pesah_spacetime_2025}, in which the routed extraction circuit helps determine the relevant noise structure. The findings also connect the implementation problem to AKP long-range noise~\cite{akp_2006}: for planar layouts with regularized algebraic kernels, the same $\alpha>2$ condition that governs two-dimensional summability also governs the aggregate pair couplings in the model. The biplanar hierarchy observed on BB72 and BB144 is compatible with BB memory proposals across multiple architecture variants---local, modular, multilayer, and neutral-atom~\cite{berthusen_local_2025,shaw_terhal_2025,strikis_berent_2023,mathews_multilayer_2025,liang_openboundaries_2025,wang_mueller_2026,poole_rydberg_2025,zhou_louvre_2025,berthusen_adaptive_2025,he_extractors_2025,yoder_tourdegross_2025}---and with the recent BB experiment on long-range-coupled superconducting hardware~\cite{wang_demo_2025}.

Corollary~\ref{cor:certified-Jkappa} bounds worst-support pair-event incidence, but translating that bound into a monotone logical-error ordering for arbitrary decoders would require additional structure, such as a stochastic dominance condition on the per-support pair-event distribution or a decoder-monotonicity guarantee.

\subsection{Limitations}

The microscopic model retains only the inter-block interaction term of the general two-block decomposition (Appendix~\ref{app:two-block}). The additive-local component $J_1(\hat P_e+\hat P_{e'})$, which contributes independent single-block faults rather than correlated pairs, is omitted. Treating both terms simultaneously with separate kernel profiles would generalize the model, but it is not pursued here. The model retains only two-block pair couplings and uses a Pauli-twirl approximation, thereby omitting multi-block correlations and coherent accumulation across rounds. The kernel family $\kappa(d)$ is phenomenological rather than derived from a device-specific electromagnetic model, so the absolute size of the geometry penalty depends on the assumed profile; public superconducting-hardware data constrain the effective crosstalk scale and decay envelope (Appendix~\ref{app:hardware}). The computational study centers on BB72 and BB144 under a single production decoder (BP+OSD); although Appendix~\ref{app:decoder} gives correlation-aware constructions for the retained model, no matched circuit-level decoder comparison is included. The logical-aware design program is implemented only for the four-column single-layer architecture. The AKP-type compatibility criterion is asymptotic.

\subsection{Future directions}

Appendix~\ref{app:decoder} shows how to build correlation-aware decoders for the retained model. Recent correlated and BB-specific decoding work~\cite{maan_correlated_2026,sahay_matching_2026} suggests that geometry-aware priors can be incorporated into production decoders. The same analytical framework also applies to modular, open-boundary, multilayer, and neutral-atom qLDPC extraction schemes with explicit routing geometry~\cite{strikis_berent_2023,liang_openboundaries_2025,mathews_multilayer_2025,poole_rydberg_2025}. On the design side, the logical-aware program can be extended to larger BB families and beyond pure-$q(L)$ objectives, and Theorem~\ref{thm:first-order-la} points to $C_D(\phi)$ as the sharper weak-correlation objective for a fixed decoder.

Several directions would strengthen the microscopic-to-retained bridge. Deriving the effective inter-block coupling matrix $J_{e,e'}(\phi)$ from a device-level Hamiltonian via Schrieffer--Wolff or black-box quantization methods would replace the phenomenological kernel $\kappa(d)$ with a controlled reduction, moving the microscopic front end from model-dependent to derived. The schedule-level retained reduction is already controlled by Theorem~\ref{thm:schedule-prop}, Corollary~\ref{cor:linearized}, and Eq.~\eqref{eq:cycle-compose}. The decoder-aware coefficient $C_D(\phi)$ from Theorem~\ref{thm:first-order-la} has been evaluated at a single operating point (Appendix~\ref{app:decoder}); extending the evaluation to the full parameter grid and to the retained MAP decoder of Appendix~\ref{app:decoder} would give a sharper design objective than $J_\kappa$ across the entire sweep.

Several additional directions remain open. Treating both interaction and additive-local terms simultaneously with independent kernel profiles would generalize the model. Multi-round coherent accumulation beyond the per-round twirl is a second extension. Experimental calibration of $\kappa(d)$ from same-platform crosstalk data would constrain the kernel quantitatively, with Appendix~\ref{app:hardware} giving a preliminary anchoring. Replacing the point-to-point closest-approach separation $d_\phi(e,e')$ with a path-integrated coupling $\int\!\!\int\kappa(|\gamma_e(s)-\gamma_{e'}(t)|)\,ds\,dt$ over the routed curves would parallel mutual-inductance and mutual-capacitance integrals in superconducting circuit design~\cite{kosen_crosstalk_2024,barrett_flux_2023}. Finally, extensions to non-CSS codes and to beyond-Pauli twirls would broaden the applicable code and noise families.

\begin{acknowledgments}
The author thanks Armands Strikis for suggesting the crossing-based noise model that initiated this work and for guidance and discussions throughout, and Tam\'as Noszk\'o for suggesting the logical-aware embedding optimization. Simulations have been performed using resources provided by the Laboratory for Scientific Computing (LSC) at the Cavendish Laboratory, University of Cambridge.
\end{acknowledgments}

\section*{Data and code availability}
The simulation pipeline, processed Monte Carlo data, and figure-generation scripts are archived at \url{https://doi.org/10.5281/zenodo.19337541} and developed at \url{https://github.com/angelodibella/works}.

\appendix

\section{Claim taxonomy and notation}
\label{app:tables}

Table~\ref{tab:claim-taxonomy} classifies the paper's claims by evidential status. Table~\ref{tab:notation} collects the main symbols.

\begin{table}[!htbp]
\caption{Claim taxonomy.}\label{tab:claim-taxonomy}
\footnotesize
\renewcommand{\arraystretch}{1.2}%
\begin{ruledtabular}
\begin{tabular}{ll}
\textbf{Status} & \textbf{Claim} \\
\colrule
Proved & General Pauli-twirl (Thm.~\ref{thm:general-twirl}) \\
& $w_{\mathrm{eff}}=|S|-\nu$ (Thm.~\ref{thm:deff}) \\
& Complete-graph $w_{\mathrm{eff}}=\lceil|S|/2\rceil$ (Cor.~\ref{cor:complete-weff}) \\
& Two-sided exposure bound (Prop.~\ref{prop:union}) \\
& Summability $\alpha>2$ (Lem.~\ref{lem:summability}) \\
\hline
Controlled & Weight-$\le 2$ truncation (Thm.~\ref{thm:schedule-prop}) \\
& AKP-type condition $J_0\tau\eta_{\max}<\eta_0^{\mathrm{blk}}$ \\
\hline
Model & Retained coefficients $b_{ij}^\sigma$ (Eq.~\eqref{eq:retained-channel}) \\
& Logical-aware objective $J_\kappa$ \\
& Linearized retained single/pair sampling model \\
\hline
Empirical & $\bbSeventyTwoPowerLERRatio{}\times$ mono/bi LER ratio, BB72, $\alpha\!=\!3$ (Sec.~\ref{subsec:exposure-ordering}) \\
& $\rho_\mathrm{S}=0.893$, 101 baseline points (Sec.~\ref{subsec:exposure-ordering}) \\
& $\rho_\mathrm{S}=0.965$ at fixed $p\!=\!10^{-3}$, 59 points (Sec.~\ref{subsec:exposure-ordering}) \\
& 64/65 matched pairs concordant (Sec.~\ref{subsec:exposure-ordering}) \\
& $26\%$ worst-case exposure reduction (Sec.~\ref{subsec:la}) \\
& $C_D$ ordering matches $J_\kappa$; pilot, one point (App.~\ref{app:decoder}) \\
\end{tabular}
\end{ruledtabular}
\end{table}

\begin{table}[!htbp]
\caption{Notation summary.}\label{tab:notation}
\footnotesize
\renewcommand{\arraystretch}{1.4}%
\begin{ruledtabular}
\begin{tabular}{ll}
$\bb{n,k,d}$ & Code with $n$ qubits, $k$ logical qubits, distance $d$ \\
$\phi$ & Routed embedding \\
$d_\phi(e,e')$ & Routed separation between gate blocks \\
$\kappa(d)$ & Proximity kernel (dimensionless, $\kappa(0)=1$) \\
$J_0,\,\tau$ & Coupling scale, gate-block duration \\
$\theta(d)$ & Dimensionless phase $\tau J_0\kappa(d)$ \\
$p(d)$ & Twirled pair-fault probability $\sin^2\theta(d)$ \\
$\eta_e^{(t)}(\phi)$ & Same-tick block exposure of gate block $e$ \\
$A_t$ & Set of simultaneously active block pairs in tick $t$ \\
$\Theta_t$ & Tickwise coupling norm, $\Theta_t^2=\sum_a\theta_a^2$ \\
$m_t$ & Collision multiplicity (per Pauli monomial) \\
$q_\kappa(d)$ & Retained sector coefficient \\
$\Pgroup_k$ & $k$-qubit Pauli group (modulo phases) \\
$\mathcal T$ & Pauli twirl \\
$\mathcal N_\phi^\sigma$ & Full Pauli-twirled data channel (sector $\sigma$) \\
$\mathcal N_{\phi,\le 2}^\sigma$ & Retained single-and-pair data channel \\
$C_\phi^\sigma$ & Weighted correlation graph (sector $\sigma$) \\
$\nu_\phi^\sigma(L)$ & Matching number on support of $L$ \\
$\Ew_\phi^\sigma(L)$ & Weighted exposure on support of $L$ \\
$\eta_{\max}(\phi)$ & Worst-case block exposure \\
$\eta_0^{\mathrm{blk}}$ & AKP gate-block threshold constant \\
$J_\kappa(\phi;\,\RX)$ & Logical-aware objective (max over $\RX$) \\
$C_D(\phi)$ & First-order decoder-aware coefficient \\
\end{tabular}
\end{ruledtabular}
\end{table}

\section{General two-block decomposition and the $J_1/J_2$ regime split}
\label{app:two-block}

This appendix records the general decomposition of any Hermitian two-block perturbation into local and inter-block pieces, identifies the two first-order regimes (additive-local and interaction), and fixes the projection formulas that assign an amplitude to each regime. The body of the paper retains only the interaction component (Assumption 2); the construction here records what that choice discards.

\subsection{Unique local--interaction split}

\begin{proposition}[General two-block decomposition]
\label{prop:two-block-decomp}
Let $\mathcal{H}_e$ and $\mathcal{H}_{e'}$ be the Hilbert spaces of two disjoint active gate blocks with dimensions $D_e$ and $D_{e'}$, and let $\hat K$ be any Hermitian operator on $\mathcal{H}_e\otimes\mathcal{H}_{e'}$. Then there exists a unique decomposition
\begin{equation}
\hat K = c\,\I + \hat A_e\otimes\I_{e'} + \I_e\otimes\hat B_{e'} + \hat C_{e,e'},
\label{eq:two-block-decomp}
\end{equation}
in which $c\in\mathbb{R}$, the local operators $\hat A_e$ and $\hat B_{e'}$ are traceless Hermitian on their respective blocks, and $\hat C_{e,e'}$ is Hermitian with vanishing partial traces,
\begin{equation}
\operatorname{tr}_{e'}\hat C_{e,e'}=0, \qquad \operatorname{tr}_e\hat C_{e,e'}=0.
\label{eq:C-partial-traces}
\end{equation}
\end{proposition}

\begin{proof}
Define
\begin{align}
c &= \frac{1}{D_e D_{e'}}\operatorname{tr}\hat K, \\
\hat A_e &= \frac{1}{D_{e'}}\operatorname{tr}_{e'}\hat K - c\,\I_e, \\
\hat B_{e'} &= \frac{1}{D_e}\operatorname{tr}_e\hat K - c\,\I_{e'}, \\
\hat C_{e,e'} &= \hat K - c\,\I - \hat A_e\otimes\I_{e'} - \I_e\otimes\hat B_{e'},
\end{align}
with $\operatorname{tr}_e$ and $\operatorname{tr}_{e'}$ the partial traces over the respective subsystems. Partial traces preserve Hermiticity, so all four pieces are Hermitian. By construction, $\operatorname{tr}\hat A_e=0$, $\operatorname{tr}\hat B_{e'}=0$, and both partial traces of $\hat C_{e,e'}$ vanish. For uniqueness, suppose
\begin{equation}
0 = \tilde c\,\I + \tilde{\hat A}\otimes\I + \I\otimes\tilde{\hat B} + \tilde{\hat C}
\end{equation}
with the same constraints. Taking the total trace gives $\tilde c=0$; taking $\operatorname{tr}_{e'}$ gives $\tilde{\hat A}=0$; taking $\operatorname{tr}_e$ gives $\tilde{\hat B}=0$; and therefore $\tilde{\hat C}=0$.
\end{proof}

Proposition~\ref{prop:two-block-decomp} separates any geometry-induced perturbation into a global phase $c$, local block responses $\hat A_e$, $\hat B_{e'}$, and an inter-block interaction $\hat C_{e,e'}$. Whether the local or the inter-block component dominates is set by the microscopic coupling mechanism.

\subsection{First-order local-field reduction}

\begin{corollary}[First-order local-field reduction]
\label{cor:local-field}
If the same-tick geometry effect is a state-independent stray field of amplitude $g(d)$ to which each block responds linearly through Hermitian operators $\hat R_e$ and $\hat R_{e'}$, then to first order in $g(d)$ the inter-block term vanishes and
\begin{widetext}
\begin{equation}
\hat K(d)=g(d)\left(\hat R_e\otimes\I_{e'}+\I_e\otimes\hat R_{e'}\right)+\;c(d)\,\I+O\!\left(g(d)^2\right).
\label{eq:local-field}
\end{equation}
\end{widetext}
Choosing the dominant response channel on each block, $\hat R_e=\hat P_e$ and $\hat R_{e'}=\hat P_{e'}$, and identifying $g(d)=J_0\kappa(d)$, gives the additive-local Hamiltonian $J_0\kappa(d)(\hat P_e+\hat P_{e'})$, the $J_1$-only limit of Eq.~\eqref{eq:two-parameter-family} below.
\end{corollary}

\begin{proof}
A state-independent field acts on each block independently. The first-order perturbation on block $e$ cannot act nontrivially on $\mathcal{H}_{e'}$, so it takes the form $\hat A_e(d)\otimes\I_{e'}$, and likewise for $e'$. Linearity in the weak field gives proportionality to $g(d)$, and Hermiticity forces the response operators to be Hermitian.
\end{proof}

\subsection{Two-parameter subfamily and its $O(\theta^2)$ consequences}

\begin{remark}
\label{rem:general-family}
Projecting a general perturbation $\hat K$ onto the chosen block channels $\hat P_e$ and $\hat P_{e'}$ yields a two-parameter subfamily
\begin{equation}
\hat K(d)=J_1(d)\!\left(\hat P_e+\hat P_{e'}\right)+J_2(d)\,\hat P_e\otimes\hat P_{e'}.
\label{eq:two-parameter-family}
\end{equation}
The remaining Pauli components of $\hat K$ are neglected because, for each gate type, only one Pauli channel propagates through the Clifford schedule to the relevant CSS sector. Other components either commute with the stabilizers or map to the opposite sector. Under the Pauli twirl, the $J_1$ (additive-local) term produces only independent single-block faults, so geometry-induced pair data faults appear from this term only at $O(\theta^4)$, through joint two-block events. The $J_2$ (interaction) term produces correlated pair faults directly at $O(\theta^2)$. The body of the paper retains the interaction term alone, which dominates the correlated-fault budget in the interaction-dominated regime. Corollary~\ref{cor:local-field} gives the opposite (additive-local) limit.
\end{remark}

The coefficients $J_1$ and $J_2$ are projections of $\hat K$ onto the chosen block channels. For any two-block perturbation $\hat K$ on $\mathcal{H}_e\otimes\mathcal{H}_{e'}$,
\begin{equation}
J_1(d) = \frac{1}{D_e D_{e'}}\operatorname{tr}\!\left[\hat K(d)\,(\hat P_e\otimes\I_{e'})\right]
\label{eq:J1-proj}
\end{equation}
and
\begin{equation}
J_2(d) = \frac{1}{D_e D_{e'}}\operatorname{tr}\!\left[\hat K(d)\,(\hat P_e\otimes\hat P_{e'})\right],
\label{eq:J2-proj}
\end{equation}
with $\operatorname{tr}$ the trace over both subsystems. Assumption~2 of the main text sets $J_1=0$ and identifies $J_2(d)=J_0\kappa(d)$. In a measured or simulated perturbation, Eqs.~\eqref{eq:J1-proj}--\eqref{eq:J2-proj} determine both amplitudes from the extracted effective perturbation and indicate which regime applies.

The four-qubit two-block motif in Appendix~\ref{app:microscopic} tests the decomposition numerically. For the pure stray-drive perturbation ($J_2=0$) the inter-block norm vanishes identically. For the mixed regime ($J_2/J_1=0.1$) it is an order of magnitude smaller than the local pieces but dominates the leading correlated-pair sector, at $O(\theta^2)$ versus $O(\theta^4)$. The simulations in the main text use the pure interaction model ($J_1=0$).

\section{Supplementary geometry figures}
\label{app:geometry}

This appendix collects the supplementary geometry figures cited in the main text.

\subsection{Toric-base placement}

Figure~\ref{fig:toric-base} gives the toric-base placement rule used in the numerical bounded-thickness construction. The routing-layer assignment is applied only after the base-plane placement is fixed.

\begin{figure}[!htbp]
\centering
\includegraphics[width=0.98\columnwidth]{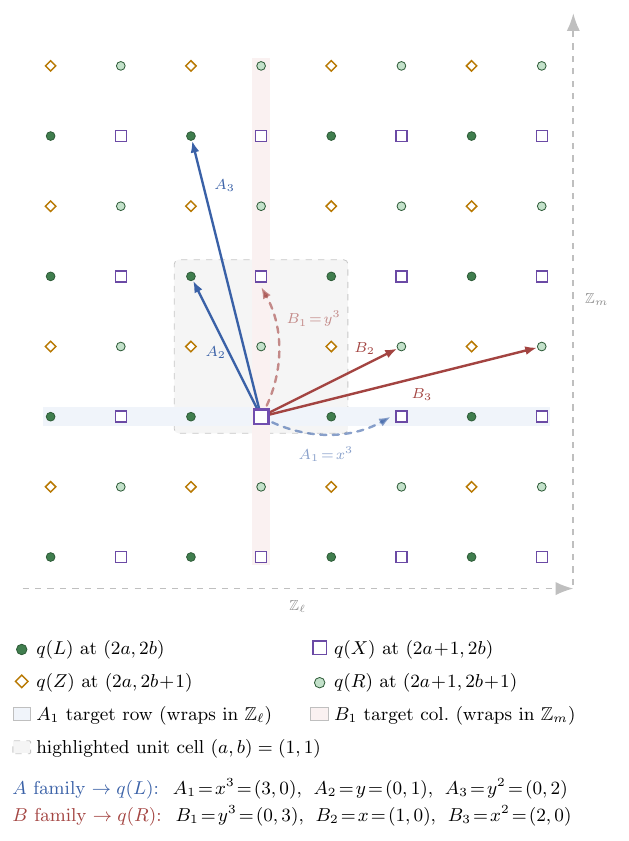}
\caption{Toric-base placement rule underlying the biplanar bounded-thickness embedding. The displayed lattice patch uses the coordinate convention $q(L):(2a,2b)$, $q(X):(2a{+}1,2b)$, $q(Z):(2a,2b{+}1)$, $q(R):(2a{+}1,2b{+}1)$. The highlighted $q(X)$ site at cell $(a,b)=(1,1)$ shows the $A$-family (blue, to $q(L)$) and $B$-family (red, to $q(R)$) shift vectors for the BB72 code. Dashed arrows indicate terms that wrap via periodic boundaries; the blue and red stripes highlight the target row and column, respectively. Routing-layer assignment is applied only after this base-plane placement has been fixed.}
\label{fig:toric-base}
\end{figure}

\section{Microscopic motif diagnostics}
\label{app:microscopic}

The two microscopic diagnostics cited in Sec.~\ref{sec:model} are presented here.

Figure~\ref{fig:two-block-norms} tests the two-block decomposition (Proposition~\ref{prop:two-block-decomp}) on a four-qubit motif (two simultaneous \textsc{cnot}s). For the pure stray-drive perturbation ($J_2=0$), the inter-block norm vanishes identically; for the mixed regime ($J_2/J_1=0.1$), it is an order of magnitude smaller than the local pieces but dominates the leading correlated-pair sector at $O(\theta^2)$ versus $O(\theta^4)$.

\begin{figure}[!htbp]
\centering
\includegraphics[width=\columnwidth]{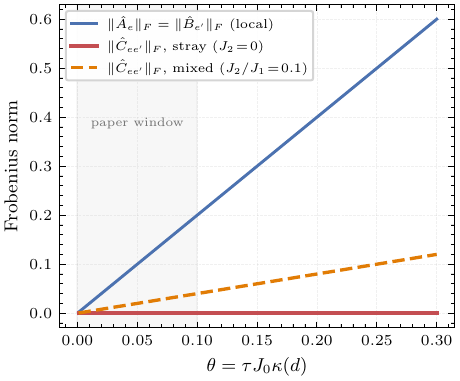}
\caption{Two-block decomposition norms from Proposition~\ref{prop:two-block-decomp} on a four-qubit motif (two simultaneous \textsc{cnot}s). The local norms $\lVert\hat A_e\rVert_F=\lVert\hat B_{e'}\rVert_F$ grow linearly with $\theta$ (blue). For the stray-drive perturbation ($J_2=0$), the inter-block norm $\lVert\hat C_{e,e'}\rVert_F$ vanishes identically (red, on the $x$-axis). For the mixed perturbation ($J_2/J_1=0.1$), $\hat C_{e,e'}$ is nonzero but $10\times$ smaller than the local pieces (orange dashed). The shaded band marks the coupling window used in the circuit-level simulations.}
\label{fig:two-block-norms}
\end{figure}

Figure~\ref{fig:weight-audit} tests the schedule-level propagation result (Theorem~\ref{thm:schedule-prop}) on a minimal one-round subcircuit (three data qubits connected to one ancilla by three consecutive \textsc{cnot}s). The ancilla is eliminated via Kraus decomposition, and the resulting data channel is Pauli-twirled using
\begin{equation}
p_P=\frac{1}{4^n}\sum_i\left|\operatorname{tr}\hat P\hat K_i\right|^2.
\end{equation}
At $\theta=0$ the baseline channel has $M_0=M_3=0.5$ and $M_1=M_2=0$, reflecting the even-parity ($ZZZ$) stabilizer projection. Panel~(a) shows the geometry-induced weight-$1$ mass tracking $\sin^2\theta$ to within $2\%$ for the stray drive, confirming faithful twirl propagation; $M_2=M_1$ by subcircuit symmetry. Panel~(b) shows $M_{\ge 3}(\theta)$ decreasing monotonically: the geometry-induced increment $\Delta M_{\ge 3}<0$, so the crosstalk generates only weight-$1$ and weight-$2$ mass in this minimal setting.

\begin{figure*}[!htbp]
\centering
\includegraphics[width=0.88\textwidth]{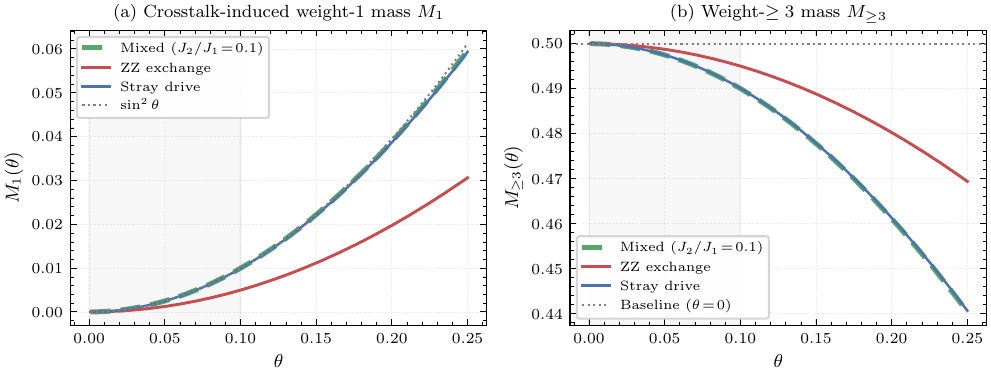}
\caption{Pauli-weight spectrum on a three-data-qubit one-round subcircuit after ancilla elimination and Pauli twirl. The baseline channel ($\theta=0$) has $M_0=M_3=0.5$ from the stabilizer projection; all geometry-induced mass appears at weights $1$ and $2$. (a)~Weight-$1$ mass $M_1(\theta)$ for three coupling types; the dotted line is $\sin^2\theta$. The weight-$2$ mass $M_2=M_1$ by the subcircuit symmetry. (b)~Total weight-$\ge 3$ mass $M_{\ge 3}(\theta)$, starting at $0.5$ (dotted) and decreasing: the geometry-induced increment $\Delta M_{\ge 3}<0$, so the crosstalk does not generate new higher-weight contributions.}
\label{fig:weight-audit}
\end{figure*}

\section{Supplementary numerical diagnostics}
\label{app:numerics}

This appendix collects the supplementary numerical diagnostics cited in the main text.

\subsection{Additional BB72 sweeps}

Figure~\ref{fig:bb72-additional} gathers three additional BB72 diagnostics: the exponential-kernel range sweep, the physical-error-rate sweep at fixed algebraic kernel, and a phase-diagram heat map of the monomial-to-biplanar logical-error-rate ratio at $p=3\times10^{-3}$. Across the sampled window, the biplanar embedding has a lower logical error rate.

\begin{figure*}[!htbp]
\centering
\includegraphics[width=\textwidth]{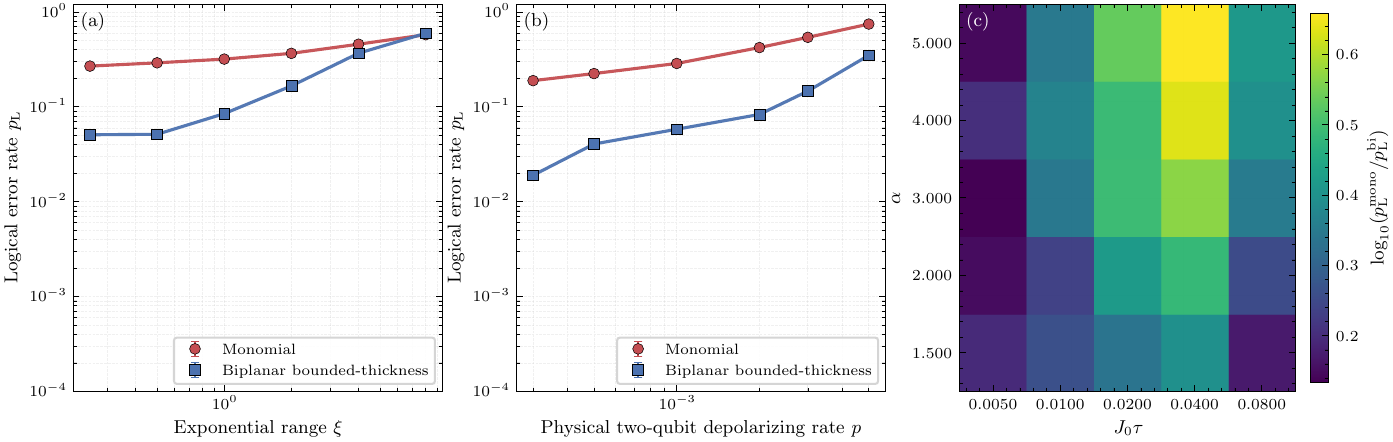}
\caption{Supplementary BB72 diagnostics in the $X$ sector. (a) Exponential-kernel sweep at fixed $J_0\tau=0.04$ and $p=10^{-3}$. (b) Physical-error-rate sweep for the regularized algebraic kernel at fixed $J_0\tau=0.04$ and $\alpha=3$. (c) Heat map of $\log_{10}(\mathrm{LER}_{\mathrm{mono}}/\mathrm{LER}_{\mathrm{bi}})$ across the available $(J_0\tau,\alpha)$ grid at $p=3\times10^{-3}$.}
\label{fig:bb72-additional}
\end{figure*}

Figure~\ref{fig:bb72-phase-diagram} shows the absolute logical error rate of both embeddings across the full $(J_0\tau,\alpha)$ plane. At fixed coupling, the monomial embedding varies little with $\alpha$, whereas the biplanar embedding suppresses logical error rate more strongly as the decay exponent moves above the summability threshold $\alpha=2$.

\begin{figure*}[!htbp]
\centering
\includegraphics[width=\textwidth]{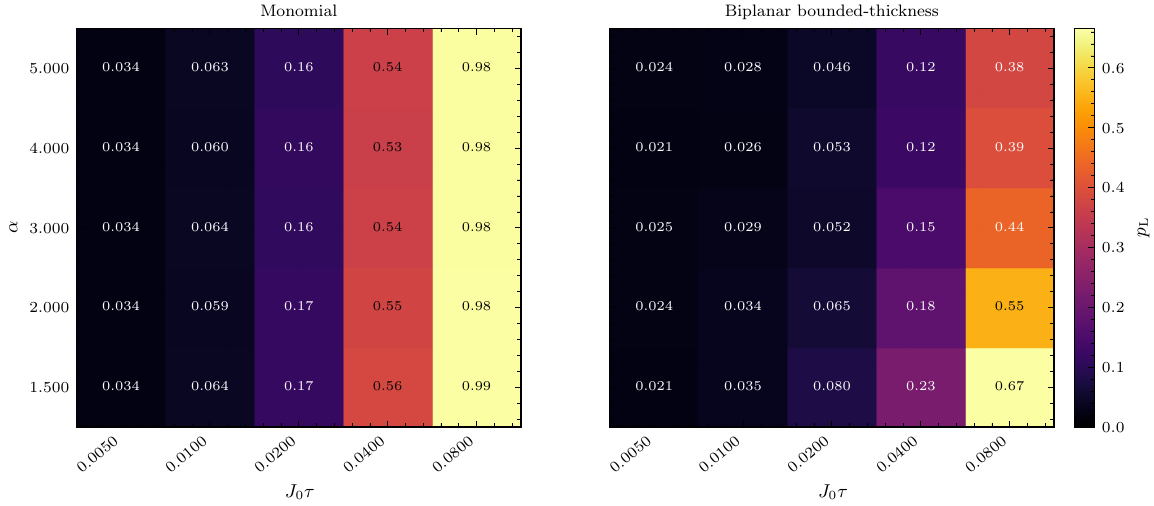}
\caption{BB72 phase-diagram heatmaps of the logical error rate across the $(J_0\tau,\alpha)$ parameter plane at $p=3\times10^{-3}$ in the $X$ sector. Left: monomial embedding. Right: biplanar embedding. Both panels use the same color scale. Cell annotations are the total logical error rate $p_L$.}
\label{fig:bb72-phase-diagram}
\end{figure*}

Figure~\ref{fig:bb72-by-kernel} separates the BB72 weighted-exposure scatter by kernel family. The crossing-kernel points provide the high-contrast diagnostic regime, whereas the positive-kernel points show the smoother exposure-controlled regime used in the main text.

\begin{figure*}[!htbp]
\centering
\includegraphics[width=\textwidth]{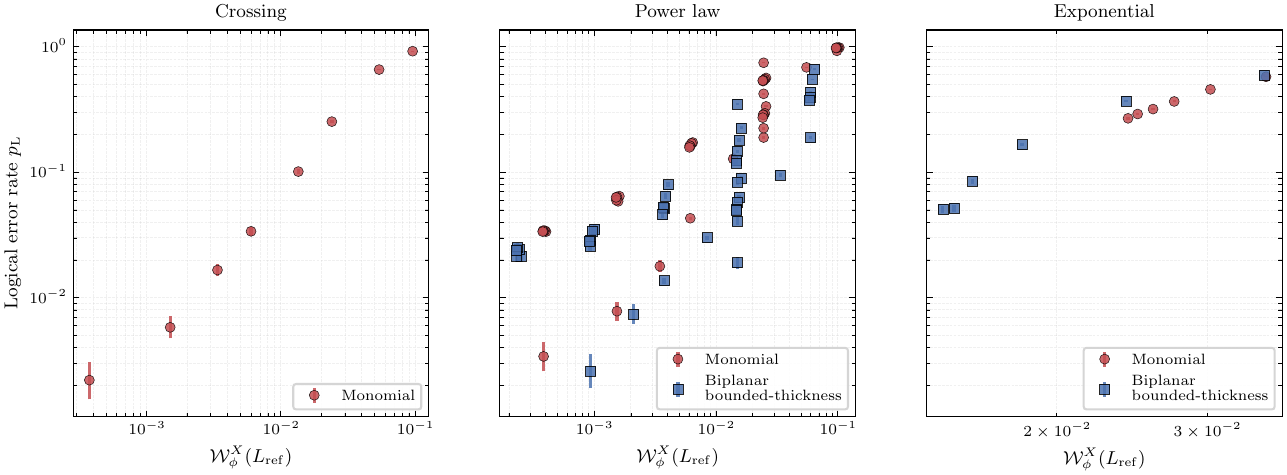}
\caption{BB72 weighted-exposure scatter split by kernel family (cf.\ the combined scatter in the main text, Fig.~\ref{fig:scatter}). Left: crossing kernel, which occupies the high-contrast validation regime. Center: power-law kernel, where both embeddings span a wide exposure range. Right: an exponential kernel that covers a narrower exposure window. Within each kernel class, the monotonic exposure--LER trend persists, confirming that the branch structure visible in the combined scatter arises from pooling different kernel families rather than from a breakdown of the exposure metric.}
\label{fig:bb72-by-kernel}
\end{figure*}

\subsection{Supporting evidence from BB90 and BB108}

Figure~\ref{fig:family-slices} gives the lightweight BB90 and BB108 $J_0\tau$ slices. They preserve the same embedding ordering as BB72 and BB144, but they are used only as supporting evidence.

The intermediate geometry audit also shows that, on BB90, a simple maximum-exposure score can be anisotropic enough that the biplanar embedding need not have lower exposure than the monomial embedding on every worst-case metric over supports, even though the sampled logical-error-rate hierarchy still favors bounded thickness. BB90 and BB108 are therefore treated as supplementary evidence.

\begin{figure*}[!htbp]
\centering
\includegraphics[width=\textwidth]{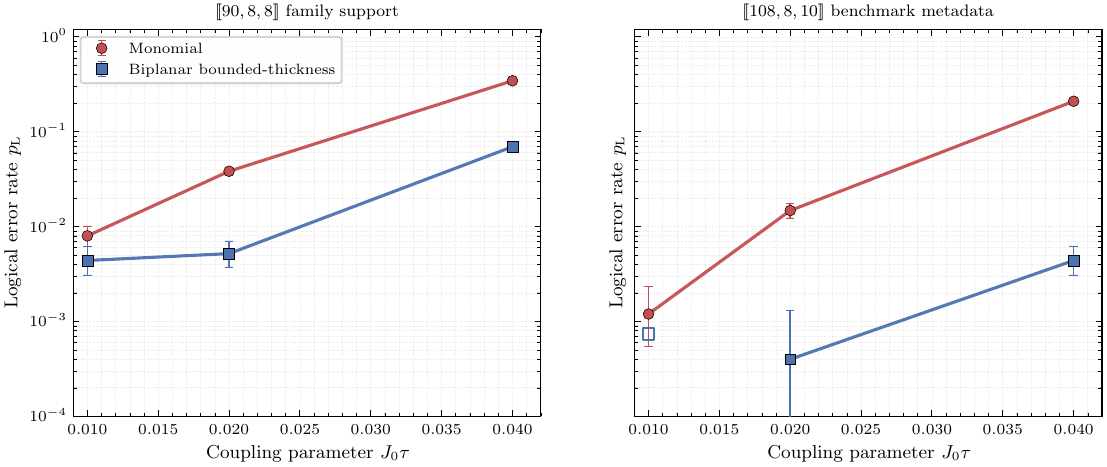}
\caption{Supporting evidence from BB90 and BB108 in the $X$ sector for the regularized algebraic kernel at $\alpha=3$ and $p=10^{-3}$. These slices support the same qualitative embedding hierarchy as the main BB72 and BB144 results.}
\label{fig:family-slices}
\end{figure*}

\subsection{Biplanar scaling diagnostic}

Figure~\ref{fig:biplanar-scaling} compares the biplanar BB72 and BB144 sweeps directly. It is included as a compact operating window diagnostic; no asymptotic conclusion is drawn from it.

\begin{figure}[!htbp]
\centering
\includegraphics[width=\columnwidth]{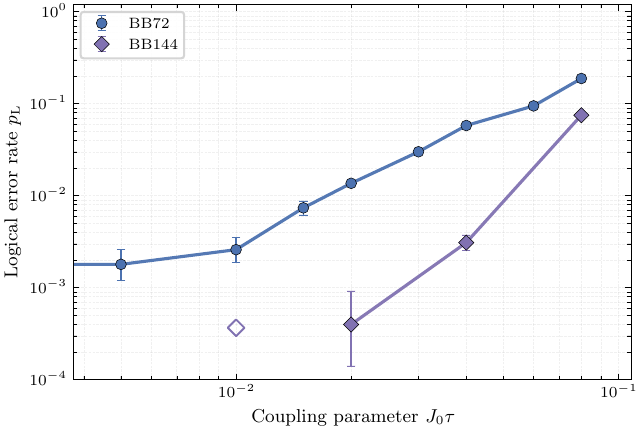}
\caption{Supplementary biplanar comparison between BB72 and BB144 under the regularized algebraic kernel at $\alpha=3$ and $p=10^{-3}$. The figure is included only as an operating-window diagnostic; no asymptotic threshold claim is extracted from it.}
\label{fig:biplanar-scaling}
\end{figure}

\subsection{Robustness checks}

Two reduced slices verify that the embedding hierarchy does not rely on tied local rates or on the choice of the $X$ sector.

\emph{Untied local rates.} Setting $p_{\mathrm{cnot}}=p_{\mathrm{prep}}=p_{\mathrm{meas}}=10^{-3}$ and $p_{\mathrm{idle}}=10^{-4}$ leaves the ordering unchanged: the logical error rates at the reference operating point are $0.250$ (monomial), $0.175$ (logical-aware), and $0.052$ (biplanar), compared with the tied-rate baseline $0.287$, $0.205$, and $0.058$.

\emph{$Z$-sector validation.} Decoding the $Z$ sector instead of the $X$ sector at the same operating point gives logical error rates $0.288$ (monomial), $0.216$ (logical-aware), and $0.059$ (biplanar), so the geometry penalty is not specific to the $X$ sector.

\section{Further logical-aware design details}
\label{app:la-details}

This appendix records two additional logical-aware details: a BB108 benchmark-metadata note and the thickness-two control calculation.

\subsection{Benchmark distance metadata versus pure-\texorpdfstring{$q(L)$}{q(L)} minimum weight}

The literature benchmark is \bb{108,8,10}, whereas the minimum weight among pure-$q(L)$ representatives relevant to the logical-aware objective is $12$. The first is a full code property; the second is a restricted-family property used by the design program of Sec.~\ref{sec:logical-aware}. We keep them separate throughout.

\subsection{Thickness-two extension}

Admissible thickness-two embeddings behave differently. Let $\mathfrak E_{\mathrm{bi}}$ be the family of bounded-thickness embeddings that preserve the BB layer split and are planar within each relevant same-round layer. Then the crossing-local objective is trivial on $\mathfrak E_{\mathrm{bi}}$.

\begin{proposition}[Triviality of the crossing-local objective on admissible bounded-thickness families]
For every $\phi\in\mathfrak E_{\mathrm{bi}}$, the implemented crossing-local correlation graph in the relevant $X$-sector rounds is empty. Hence $J_{\times}(\phi;\,\RX)=0$ for all $\phi\in\mathfrak E_{\mathrm{bi}}$.
\end{proposition}

\begin{proof}
By definition of admissibility, each relevant same-round routed layer is planar, so no same-round projected route pairs cross. The crossing kernel, therefore, assigns zero weight to every candidate pair edge.
\end{proof}

The thickness-two objective on $\mathfrak E_{\mathrm{bi}}$ is therefore a distance-decay quantity such as the upper bound in Eq.~\eqref{eq:bi-surrogate-proof}. The present audit shows only modest additional gains from optimizing within the bounded-thickness family, so the main text keeps the default biplanar embedding fixed.

\section{Decoder mismatch and correlation-aware decoding}
\label{app:decoder}

The correlation-aware BP+OSD (CBP+OSD) construction below is formulated for the retained model rather than for the full circuit-level channel of the multi-round simulation, so no decoder-performance comparison is made here.

\subsection{Decoder-mismatch theorem}

Take $H=(1,1,1,1)$ over $\F$ and suppose the measured syndrome is $1$. Under an independent and identically distributed (iid) prior with bit-flip probability $p<1/2$, maximum-likelihood (ML) decoding prefers any weight-$1$ error to any weight-$3$ error. Under the correlated prior
\begin{equation}
\mu_{\mathrm{corr}}(e)\propto \mu_{\mathrm{iid}}(e)\exp\!\left[J\sum_{i<j}e_i e_j\right],
\end{equation}
the ordering can be reversed.

\begin{theorem}[Factorized-prior ML and correlated-prior MAP can disagree]
\label{thm:mismatch}
For the single-check code above, if
\begin{equation}
J>\frac{2}{3}\log\left(\frac{1-p}{p}\right),
\end{equation}
then correlated-prior MAP decoding chooses a weight-$3$ error while factorized-prior ML decoding chooses a weight-$1$ error.
\end{theorem}

\begin{proof}
Any odd-weight error has syndrome $1$. Under the iid prior, the log-probability difference between a weight-$3$ error and a weight-$1$ error is
\begin{equation}
\log\frac{p^3(1-p)}{p(1-p)^3}=2\log\frac{p}{1-p}<0,
\end{equation}
so the weight-$1$ configuration is preferred. Under the correlated prior, the weight-$3$ configuration gains an additional factor $e^{3J}$ because it has three occupied pairs, whereas the weight-$1$ configuration has none. The weight-$3$ error is preferred if and only if
\begin{equation}
2\log\frac{p}{1-p}+3J>0,
\end{equation}
which is the stated condition.
\end{proof}

\subsection{Exact augmented decoding for the retained single-and-pair model}

Let $x_i\in\{0,1\}$ denote a retained single fault on data location $i$, and let $y_a\in\{0,1\}$ denote a retained pair fault on edge $a=(i_a,j_a)$ of the retained correlation graph. Collect them into the latent fault-location vector
\begin{equation}
z=\left(x_1,\ldots,x_n,y_1,\ldots,y_m\right)^\top.
\end{equation}
If $F$ maps latent single and pair locations to data errors, then the observed syndrome obeys
\begin{equation}
s=HFz=\widetilde H z,
\qquad
\widetilde H:=HF.
\end{equation}
Under an \emph{independent retained-location prior},
\begin{equation}
\Pr(z)\propto \prod_i u_i^{x_i}\left(1-u_i\right)^{1-x_i}\prod_a v_a^{y_a}\left(1-v_a\right)^{1-y_a}.
\end{equation}
MAP decoding then reduces to weighted decoding on the augmented matrix.

\begin{theorem}[Exact reduction to weighted decoding]
\label{thm:augmented}
Under the independent retained-location prior, the MAP estimate is
\begin{equation}
\hat z=\argmin_{\widetilde H z=s}\left(\sum_i \lambda_i x_i+\sum_a \lambda_{i_aj_a}y_a\right),
\label{eq:augmented-map}
\end{equation}
where
\begin{equation}
\lambda_i=\log\frac{1-u_i}{u_i},
\qquad
\lambda_{i_aj_a}=\log\frac{1-v_a}{v_a}.
\end{equation}
The corresponding data estimate is $\hat e=F\hat z$.
\end{theorem}

\begin{proof}
Taking the negative logarithm of the independent retained-location prior yields the linear objective in Eq.~\eqref{eq:augmented-map} up to an additive constant. The syndrome constraint is $\widetilde H z=s$. Therefore, MAP decoding is equivalent to the stated constrained weighted minimization.
\end{proof}

Without pair variables, Theorem~\ref{thm:augmented} reduces to ordinary weighted decoding and hence to standard BP+OSD. If one truncates the retained graph by discarding sufficiently weak pair edges, the same theorem applies to the truncated retained model.

One can also work directly on the original data-error bits with the approximate pairwise-Ising correction,
\begin{equation}
\mu_{\mathrm{pair}}(e)\propto \mu_{\mathrm{iid}}(e)\exp\left(\sum_{i<j}J_{ij}e_i e_j\right).
\end{equation}
For a single pair factor coupling bits $i$ and $j$, the outgoing log-likelihood-ratio correction is
\begin{equation}
\Lambda_{f_{ij}\to i}(\ell)=\log(1+e^{-\ell})-\log(1+e^{J_{ij}-\ell}),
\end{equation}
where the incoming log-likelihood ratio from bit $j$ is
\begin{equation}
\ell=\log\frac{\Pr(e_j=0)}{\Pr(e_j=1)}
\end{equation}
(positive $\ell$ favours $e_j=0$). This approximation is useful algorithmically but is not needed for the results reported here.

Theorem~\ref{thm:mismatch} shows that factorized-prior ML and correlated-prior MAP decoding can disagree even in a four-bit example. Factorized-prior BP+OSD is therefore an appropriate baseline for the retained simulations, but it is not generally MAP-optimal when pair correlations are explicitly retained.

\subsection{Decoder-aware first-order refinement}

The objective $J_{\kappa}$ is intrinsic to the retained geometry model and does not depend on a decoder. In the weak-correlation regime, one can sharpen the design problem by expanding the logical failure probability to first order in the geometry-induced pair strength. Fix an embedding $\phi$ and a deterministic decoder $D$. Let $\xi$ collect all retained fault variables independent of a coupling parameter $\lambda$, and let $A(\phi)$ be the set of retained pair locations. For each $a\in A(\phi)$, let $Y_a\sim \mathrm{Bernoulli}(\lambda\bar v_a(\phi))$ independently, and let $F_D^{\phi}(\xi,Y)\in\{0,1\}$ indicate decoder failure.

\begin{theorem}[Weak-correlation first-order ordering]
\label{thm:first-order-la}
For every fixed embedding $\phi$, deterministic decoder $D$, and $\lambda\in[0,\,(\max_a \bar v_a(\phi))^{-1}]$,
\begin{equation}
p_L^D(\phi,\lambda)
:=
\mathbb E\left[F_D^{\phi}(\xi,Y)\right]
=
p_{L,0}^D(\phi)+\lambda C_D(\phi)+O(\lambda^2),
\end{equation}
where $p_{L,0}^D(\phi):=\mathbb E[F_D^{\phi}(\xi,0)]$, $\Delta_a^D(\phi):=\mathbb E[F_D^{\phi}(\xi,e_a)-F_D^{\phi}(\xi,0)]$, and
\begin{equation}
C_D(\phi):=\sum_{a\in A(\phi)} \bar v_a(\phi)\,\Delta_a^D(\phi).
\end{equation}
If two embeddings $\phi$ and $\phi'$ have the same baseline $p_{L,0}^D$ and satisfy $C_D(\phi')<C_D(\phi)$, then $p_L^D(\phi',\lambda)<p_L^D(\phi,\lambda)$ for all sufficiently small $\lambda>0$.
\end{theorem}

\begin{proof}
Expanding in $\lambda$: the zero-pair probability is $1-\lambda\sum_a\bar v_a+O(\lambda^2)$ and each single-pair probability is $\lambda\bar v_a+O(\lambda^2)$; all multi-pair configurations are $O(\lambda^2)$. Collecting terms gives the stated expansion.
\end{proof}

A pilot finite-difference estimate at $(p,\alpha)=(10^{-3},3)$ using BP+OSD with $10{,}000$ shots gives $C_D\approx 0.82$ (monomial), $0.45$ (logical-aware), and $0.26$ (biplanar), matching the intrinsic exposure ordering with a sharper monomial-to-biplanar ratio ($3.2\times$ for $C_D$ versus $1.8\times$ for $J_\kappa$).

\section{Public-data-informed priors for \texorpdfstring{$\kappa$}{kappa}}
\label{app:hardware}

The kernel $\kappa(d)$ and coupling scale $J_0$ are free parameters of the model. This appendix uses publicly available superconducting-qubit crosstalk data to verify that the paper's operating window is hardware-plausible and that the embedding hierarchy persists across the resulting kernel range. The public datasets constrain an effective geometry-dependent crosstalk scale and decay envelope; they do not by themselves distinguish the inter-block interaction mechanism ($\hat P_e\otimes\hat P_{e'}$) used in the simulations from other crosstalk channels such as additive-local stray fields.

\subsection{Amplitude scale}

Kosen et al.\ report xy-drive crosstalk on a 25-qubit flip-chip processor with average values $-39.4\pm 3.7$\,dB and $-37.4\pm 3.9$\,dB across two device variants, and a worst-case value of $-27$\,dB~\cite{kosen_crosstalk_2024}. Interpreting each as a spurious rotation during a $\pi/2$ target pulse via the proxy-angle map
\begin{equation}
\theta_{\mathrm{xtalk}}=\frac{\pi}{2}\,10^{x_{\mathrm{dB}}/20}
\end{equation}
gives
\begin{equation}
\theta_{\mathrm{avg}}\approx 0.017\text{--}0.021\,\mathrm{rad},
\qquad
\theta_{\mathrm{worst}}\approx 0.070\,\mathrm{rad}.
\end{equation}
Separately, the Kunlun BB processor~\cite{wang_demo_2025} reports simultaneous-CZ error $0.98\%$ versus isolated-CZ error $0.73\%$; mapping the excess $\Delta p\approx 0.25\%$ through $\theta\approx\sqrt{\Delta p}$ gives a proxy angle of $\theta\approx 0.050$\,rad. Both estimates fall within the paper's sweep window $J_0\tau\in[0.02,\,0.08]$, so the explored coupling range is consistent with public superconducting-hardware crosstalk data.

\subsection{Kernel shape and decay length}

Barrett et al.\ fit DC flux crosstalk on a 16-qubit flip-chip array to the shifted reciprocal law $c(d)=100/(ad+1)+c_0$ with $a=178.2$\,mm$^{-1}$ and $c_0=0.264\%$~\cite{barrett_flux_2023}. Normalizing to $\kappa_{\mathrm{DC}}(d)=c(d)/c(0)$ and converting to the paper's pitch units via a physical pitch $\delta\in[0.2,\,0.5]$\,mm yields
\begin{equation}
\kappa_{\mathrm{DC}}(1\,\text{pitch})\approx 0.01\text{--}0.03.
\end{equation}
Fitting an exponential $e^{-d/\xi}$ to this one-pitch value gives $\xi\approx 0.25$--$0.30$ pitch units. For fast-flux pulses the same paper reports crosstalk roughly $100\times$ smaller in amplitude, further reducing the effective coupling scale but not directly constraining the normalized decay length $\xi$ without additional shape assumptions. Conversely, Kosen's xy-drive trend of $-1$\,dB/mm corresponds to an amplitude attenuation length $\xi_{\mathrm{mm}}=20/({\ln 10})\approx 8.7$\,mm, giving $\xi\approx 17$--$43$ pitch units for $\delta\in[0.2,\,0.5]$\,mm; drive leakage therefore decays much more slowly than flux leakage.

Aguila et al.\ demonstrate that active Z-line compensation can reduce flux crosstalk from $56.5$\textperthousand{} to $0.13$\textperthousand{}~\cite{aguila_flux_2025}, a factor of $\sim\!400$, providing a lower bound on the compensated-floor regime.

\subsection{Comparison with existing simulation data}

Figure~\ref{fig:proxy-kernel} shows dedicated BB72 Monte Carlo results at $J_0\tau=0.04$ and $p=10^{-3}$ for the exponential kernel at four decay lengths: $\xi=0.25$ (flux-like), $1.0$, $4.0$, and $8.0$ pitch units. These cover the flux-like end of the hardware-informed range and extend toward, but do not fully reach, the drive-like regime ($\xi\gtrsim 17$). At $\xi=0.25$, the biplanar embedding has a logical error rate of $0.051$ versus $0.287$ for monomial (a factor of $5.6$); at $\xi=1.0$, the ratio is $3.3$. As $\xi$ increases toward the drive-like regime, the kernel flattens, and nearly all pairs contribute equally regardless of layout, so the gap narrows: at $\xi=8$ the two embeddings give comparable logical error rates ($\sim\!0.6$). The embedding hierarchy holds throughout the simulated range and is strongest in the steep-decay regime, where layer separation has the largest geometric effect.

\begin{figure}[!htbp]
\centering
\includegraphics[width=\columnwidth]{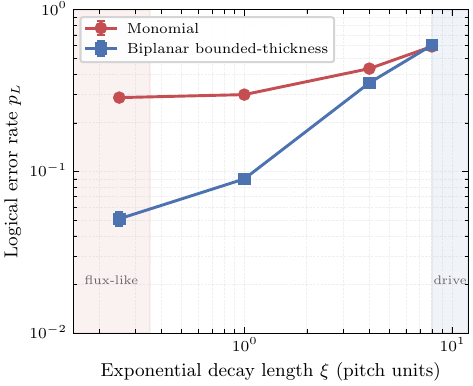}
\caption{BB72 logical error rate at $J_0\tau=0.04$ and $p=10^{-3}$ as a function of exponential decay length $\xi$. The shaded bands mark the hardware-informed regimes: flux-like ($\xi\approx 0.25$, from Barrett DC-flux data) and the lower edge of the drive-like range ($\xi\gtrsim 17$, from Kosen xy-crosstalk data; the simulated range extends to $\xi=8$). The bounded-thickness embedding outperforms the monomial embedding across the steep-decay range and converges at large $\xi$ where the kernel is nearly flat. $2000$ shots per point.}
\label{fig:proxy-kernel}
\end{figure}

\bibliography{refs}

\end{document}